\documentclass[a4paper,USenglish,cleveref, autoref, thm-restate]{lipics-v2021}
\usepackage{graphicx} 
\usepackage[normalem]{ulem}
\usepackage{xcolor,cancel}
\usepackage{tikz}

\usepackage{amsmath}
\usepackage{ulem}
\usepackage{amsthm}
\usepackage{amssymb}
\usepackage{apxproof}
\usepackage{centernot}
\usepackage{amsfonts}
\usepackage{algorithm2e}
\usepackage{float}
\usepackage{extpfeil}

\newcommand{\magenta}[1]{\textcolor{magenta}{#1}}
\newcommand{\ret}{\mathtt{return}}

\newcommand{\inv}{\mathtt{invoke}}
\newcommand{\W}{\mathtt{W}}
\newcommand{\R}{\mathtt{R}}
\newcommand{\loc}{\approx}
\newcommand{\X}{\mathrm{X}}
\newcommand{\Op}{\mathrm{O}}
\newcommand{\F}{\mathtt{F}}

\newcommand{\sa}{\mathsf{a}}
\newcommand{\sbb}{\mathsf{b}}
\newcommand{\Xa}{\mathrm{X} \sa}
\newcommand{\Y}{\mathrm{Y}}
\newcommand{\Yb}{\mathrm{Y} \sbb}

\newcommand{\RfM}{\mathtt{RfM}}
\newcommand{\RfB}{\mathtt{RfB}}

\newcommand{\ob}{\overset{ob}\rightsquigarrow}
\newcommand{\fb}{\circlearrowleft}

\newcommand{\lREAD}{\mathsf{Read}}
\newcommand{\lWRITE}{\mathsf{Write}}

\newcommand{\node}[1]{\langle#1\rangle}

\newcommand{\raissa}[1]{\magenta{#1}}

\newcommand{\prop}{\mathtt{prop}}

\newcommand{\shift}[1]{\mathtt{shift}_\Delta[#1]}
\newcommand{\hatm}{\hat{m}}
\newcommand{\Past}{\mathtt{Past}}
\newcommand{\Pastrp}{{\mathtt{Past}_r^+}}

\newcommand{\Ag}{\mathrm{Ag}}
\newcommand{\ijj}{\{i,j\}}
\usepackage{mathpartir}
\usepackage[inline]{enumitem}
\usepackage{xspace}
\usepackage{array}
\usepackage{multicol}
\usepackage{makecell}

\usepackage[noend]{algpseudocode}

\usepackage{float}

\floatstyle{ruled}
\newfloat{Example}{t}{example}

\hypersetup{
  colorlinks=true,breaklinks,
  linkcolor=green!50!black,
  citecolor=blue
}

\AtEndPreamble{%
\newtheorem{notation}[theorem]{Notation}
\newtheorem*{notation*}{Notation}
}

\newcommand{\citet}[1]{{\bf \color{red}Fix citet }\cite{#1}}

\DeclareSymbolFont{Shuffle}{U}{shuffle}{m}{n}
\DeclareFontFamily{U}{shuffle}{}
\DeclareFontShape{U}{shuffle}{m}{n}{%
  <-8>shuffle7%
  <8->shuffle10%
}{}
\DeclareMathSymbol\shuffle{\mathbin}{Shuffle}{"001}
\DeclareMathSymbol\cshuffle{\mathbin}{Shuffle}{"002}

\theoremstyle{definition}

\renewcommand{\paragraph}[1]{\smallskip\noindent{\bfseries #1}}

\newcommand{\emptyseq}{\varepsilon}

\newcommand{\set}[1]{\{{#1}\}}

\newcommand{\size}[1]{|{#1}|}
\newcommand{\stt}{\; | \;}

\newcommand{\tup}[1]{{\langle{#1}\rangle}}

\newcommand{\restrict}[2]{{#1}|_{#2}}

\newcommand{\maketil}[1]{{#1}\ldots{#1}}
\newcommand{\til}{\maketil{,}}

\newcommand{\proc}{i}

\newcommand{\act}{\mathsf{a}}
\newcommand{\Act}{{Act}}
\newcommand{\buf}{\mathsf{buf}}
\newcommand{\evt}{e}
\newcommand{\var}{x}

\newcommand{\Val}{{\mathsf{Val}}}
\newcommand{\Var}{{\mathsf{Var}}}
\newcommand{\Proc}{\Pi}

\newcommand{\Write}{\texttt{W}}
\newcommand{\Fence}{\texttt{F}}
\newcommand{\RMW}{\mathtt{RMW}}

\newcommand{\ack}{{\normalfont\texttt{ack}}}

\newcommand{\vexp}{v_\textsf{exp}}
\newcommand{\vnew}{v_\textsf{new}}

\newcommand{\ev}[2]{{#1}\,\text{:}\,{#2}}

\newcommand{\snapshoto}{\mathsf{Snapshot}}
\newcommand{\updateop}{{\rm update}}
\newcommand{\updateopp}[1]{\updateop(#1)}
\newcommand{\scanop}{{\rm scan}}

\newcommand{\noop}{\bot}

\nolinenumbers
\hideLIPIcs
\newtheorem*{theorem*}{Theorem}

\newcommand{\Mem}{\mathsf{m}}
\newcommand{\invoke}{\mathtt{invoke}}

\crefname{algocf}{algorithm}{algorithms}
\Crefname{algocf}{Algorithm}{Algorithms}
\title{Time, Fences and the Ordering of Events in TSO}

\author{Ra\"issa Nataf}{Technion, Israel} {raissa.nataf@cs.technion.ac.il}{https://orcid.org/0009-0003-1127-754X}{}

\author{Yoram Moses}{Technion, Israel} {moses@ee.technion.ac.il}{https://orcid.org/0000-0001-5549-1781}{}

\authorrunning{R.Nataf, Y.Moses}
\Copyright{Ra\"issa Nataf and Yoram Moses}
\begin{document}
\maketitle
\begin{abstract}
 The Total Store Order (TSO) is arguably the most widely used relaxed memory model in multiprocessor architectures, widely implemented, for example in Intel's x86 and x64 platforms. 
    It allows processes to delay the visibility of writes through store buffering. While this supports hardware-level optimizations and makes a significant contribution to multiprocessor efficiency, it complicates reasoning about correctness, as executions may violate sequential consistency. Ensuring correct behavior often requires inserting synchronization primitives such as memory fences ($\F$) or atomic read-modify-write ($\RMW$) operations, but this approach can incur significant performance costs.
    In this work, we develop a semantic framework that precisely characterizes when such synchronization is necessary under TSO. We introduce a novel TSO-specific {\em occurs-before} relation, which adapts Lamport’s celebrated happens-before relation from asynchronous message-passing systems to the TSO setting. Our main result is a theorem that proves that the only way to ensure that two events that take place at different sites are temporally ordered is by having the execution create an occurs-before chain between the events. 
    By studying the role of fences and $\RMW$s in creating occurs-before chains, we are then able to capture cases in which these costly synchronization operations are unavoidable. 
        Since proper real-time ordering of events is a fundamental aspect of consistency conditions such as Linearizability, our analysis provides a sound theoretical understanding of essential aspects of the TSO model. In particular, we are able to 
        generalize  prior lower bounds for linearizable implementations of shared memory objects such as registers or snapshots. 
        Our results capture the structure of information flow and causality in the TSO model by extending the standard communication-based reasoning from asynchronous systems to the TSO memory model.
\end{abstract}
\newcommand{\longrightarrowtail}{\relbar\joinrel\rightarrowtail}

\section{Introduction}

Modern multiprocessors rely on relaxed memory models to improve performance through techniques such as store buffering and out-of-order execution. Among these models, \emph{Total Store Order (TSO)}---used, for example, in Intel’s x86 and x64 architectures---is perhaps the most widely deployed. TSO improves efficiency by allowing writes to be temporarily stored in local buffers, deferring their visibility to other processors. However, this optimization comes at a conceptual cost: it breaks \emph{sequential consistency} (SC) \cite{lamport79seq}, making it significantly harder to reason about program correctness. A program that is correct under SC may behave incorrectly under TSO due to the delayed visibility of writes.
Therefore, to enforce correctness under TSO, programmers resort to synchronization primitives such as \emph{memory fences} ($\F$) and \emph{read-modify-write} ($\RMW$) actions. These primitives ensure memory visibility and ordering by flushing store buffers, or enforcing atomic access. While effective, they also inhibit hardware-level optimizations and limit concurrency, introducing significant performance overhead~\cite{LawsOfOrder2011, LahavCPP11, VafeiadisNardelli2011}. As a result, understanding \emph{when synchronization is necessary} has become a central challenge in the study of weak memory models.

Prior work has approached this challenge from both practical and theoretical angles. Automated tools have been developed to minimize or eliminate fences while preserving program correctness~\cite{DontSitOnTheFence, VafeiadisNardelli2011}, and several impossibility results have shown that synchronization is unavoidable in certain contexts. Notably, Castañeda et al.~\cite{WeakMemDISC24} proved that \emph{linearizable implementations} of shared objects under TSO must use fences or $\RMW$s in some executions. Yet, while existing results identify scenarios where synchronization is required, their methods do not reveal \emph{why} it is required. Our goal in this paper is to provide a semantic explanation that exposes the causal structure underlying TSO executions and clarifies the origin of synchronization constraints.
\subsection*{Our Approach and Contributions}

In this work, we propose a \emph{semantic framework} that allows us to precisely characterize when synchronization is required under TSO. At the heart of our framework is a novel \emph{TSO-specific occurs-before relation}, denoted $\ob$, which generalizes Lamport’s happens-before relation from asynchronous message-passing systems and adapts it to the TSO setting.
The happens-before relation has been considered in different models, including TSO \cite{AlglaveMSS10}.
In this work, we introduce a new definition for the analogue of happens-before in TSO, which we call the \emph{occurs-before} relation. One of the differences is that our definition stems from an operational model of TSO, whereas existing definitions are formulated using declarative models. As a result, our approach operates at a lower level of abstraction, making it easier to derive concrete implementation constraints.


Based on the new causal relation we state and prove the \emph{Delaying the Future (DtF) Theorem} (\Cref{thm:DtFTSO}), which shows an intimate connection between the occurs-before relation and the ability to switch the timing of events and operations in TSO. This result serves as a central tool in our analysis and formalizes the following principle: If an event is not in the causal past (under $\ob$) of another, then it can be delayed arbitrarily without any process noticing the change.
This theorem gives rise to a powerful indistinguishability argument: if no $\ob$ chain exists between two events, their relative order can be reversed without affecting any process’s local view. 
The theorem serves as a foundation for proving when the $\F$ or $\RMW$ synchronization actions are unavoidable under TSO. It allows us to precisely characterize execution patterns that inherently require synchronization, both in general and for \emph{linearizable implementations} of shared memory objects.
In particular, our framework explains, at a semantic level, why synchronization is required in the implementation of various shared memory objects—not just in specific instances, but because of the fundamental limitations on information flow imposed by TSO.

The current paper can be viewed as performing an analogous analysis for TSO to that of~\cite{CommRequirDISC24}. Our analysis highlights the role of synchronization primitives, extending lower bounds on the use of synchronization primitives in TSO.
Practically, such lower bounds identify the synchronization mechanisms implementations must employ and indicate when further attempts to remove them would be futile.
Our framework builds a conceptual bridge between memory models and asynchronous systems, offering new tools for the analysis  of synchronization in the TSO model.

\section{Related Work}
Our work is connected to several threads of research on action reordering, causality, and the necessity of synchronization in distributed systems. 
Lamport’s happens-before relation~\cite{Lam78causal}  plays a fundamental role in the study of asynchronous distributed computing, used, for example, in the design of vector clocks, race detection, causal memories \cite{mattern1989virtual,flanagan2009fasttrack,ahamad1995causal} as well as in and many other applications. 
In a precise sense, happens-before captures all of the information about timing that processes can obtain in standard asynchronous systems.
The notion of timing is essential in concurrent systems. Consider for instance Linearizability \cite{HerlihyLineari}, the gold standard correctness criterion for concurrent implementations of shared objects. Informally, an object implementation is linearizable if in each one of its executions, operations appear to occur instantaneously, in a way that is consistent with the execution and the object's specification. I.e., the actions performed and values observed by processes depend on the real-time ordering of non-overlapping operations. However, processes do not have direct access to real time in the asynchronous setting, and this makes  satisfying linearizability especially challenging.
The only way processes can obtain information about the real-time order of events in asynchronous message-passing systems is via {\em message chains} ({\em cf.}\ Lamport's {\em happens before}  relation \cite{Lam78causal}).

The happens-before relation has been considered in different models, including TSO \cite{AlglaveMSS10}.
In this work, we introduce a new definition for the analogue of happens-before in TSO, which we call the \emph{occurs-before} relation. One of the differences is that our definition stems from an operational model of TSO, whereas previous definitions such as in \cite{AlglaveMSS10} are formulated using declarative models. As a result, our approach employs a lower level of abstraction, making it easier to derive concrete implementation constraints.

Our work can be viewed as a weak-memory counterpart to a recent paper by Nataf and Moses~\cite{CommRequirDISC24} in which they use the happens-before relation to obtain 
a theorem called Delaying the Future that captures the relationship between real-time and message chains, and forms an effective  tool for the study of linearizable implementations in asynchronous systems. 
They show, for example,  that in linearizable implementations of registers, two operations (read or write) on different values must typically be connected by message chains, implying communication costs and providing insight into the structure of protocols implementing such registers.
In TSO, communication is done through memory rather than messages, with all the subtleties that  deferring of write to the memory implies. Our Delaying the Future here (\Cref{thm:DtFTSO}) is the TSO version of the DtF theorem of \cite{CommRequirDISC24}. As a consequence, many of their results about the necessity of message chains in linearizable implementations of objects apply in TSO, with message chains being replaced by~$\ob$ chains.

Attiya et al.~\cite{LawsOfOrder2011} and Castañeda et al.~\cite{WeakMemDISC24} are most directly related to our TSO results. Both works establish conditions on when  
linearizable implementations require the use of synchronization operations. In~\cite{LawsOfOrder2011}, covering arguments are used to show that precise communication patterns involving reads and writes are unavoidable in implementing classic and widely used specifications. Their results apply to objects whose methods are strongly non-commutative—such as sets, queues, or stacks—but not to registers.
The more recent \cite{WeakMemDISC24} proves a mergeability theorem for TSO traces and  applies it to obtain results about objects with one-sided non commutative methods such as registers. They show
 that \emph{linearizable} TSO implementations of objects, including registers,  must use fences or $\RMW$s in some executions.  
In \cite{CommRequirDISC24}, a DtF theorem is proved for asynchronous message-passing systems. It is then used to prove the need to construct message chains between operations in linearizable register implementations in that model. 

\section{Model and Preliminary Definitions}\label{sec:model}
\subsection{The Basic TSO Model}\label{sec:tsomodel}
We consider a TSO model based on the operational definition of TSO given in \cite{WeakMemDISC24}. Previous works proposed (slightly different) operational models for TSO, e.g., \cite{OwensTPHOLS09,BoudolPOPL09,burckhardt2008effective}.
It consists of a set $\Pi=\{1,\ldots,n\}$ of $n$ processes  and a finite set~$\Var$ of variables. 
A basic TSO state is a pair $\sigma=\tup{\Mem,\buf}$, where
$\Mem \in\Var \to \Val$ describes the main memory and
$\buf \in\,\Proc \to (\Var \times \Val)^*$ assigns a queue called a \emph{store buffer} to every
process.
Processes can perform actions from the set $\{\R[x], \W[x,v], \F, \RMW[x,\vexp,\vnew],\noop\}$. 
These correspond, respectively, to reading the value of~$x$---which will return a value $v\in\Val$, writing the value $v$ in~$x$, performing a \emph{fence}, performing a read-modify-write action, and (added simply for our convenience) performing a null action. To account for propagating values from the store buffer to variables in memory, we associate with each $i\in\Pi$ a ``local dispatcher'' component $d_i$ with a single action $\prop$. 
We denote $\Ag\triangleq \Pi\cup \{d_i\}_{i\in \Pi}$ and call its elements \emph{agents}.
Reads, writes and null actions are always enabled, while fences, $\RMW$ and $\prop$ actions can take place only if specific preconditions hold. We think of an action taking place as being an event. It is convenient to consider two types of read events: $\RfB(x,v)$ and $\RfM(x,v)$. The former is when the value of~$x$ is read from the store buffer, and the latter is when the value of $x$ is read from the variable $x\in\Mem$ in memory.
In both instances, the value returned by the read action is~$v\in\Val$. 
Using $\restrict{\beta}{\var}$  to denote the restriction of a store buffer $\beta$ to pairs
of the form $\tup{\var,\_}$, the preconditions and effects of the event of a single action being performed 
in a TSO state $\sigma=\tup{\Mem,\buf}$ are defined as follows: \\[1ex]
\begin{mathpar} 
  \inferrule[write]
  {\evt=\ev{\proc}{\Write(\var,v)} \\\\
  \buf' = \buf[\proc \mapsto \buf(\proc) \cdot \tup{\var, v}]}
  {\tup{\Mem, \buf} \xrightarrow{\evt} \tup{\Mem, \buf'}}
  \and
    \inferrule[read-from-buffer]
  {\evt=\ev{\proc}{\RfB(\var,v)} \\\\ 
      \restrict{\buf(\proc)}{\var} = \_ \cdot \tup{\tup{\var,v}}}
  {\tup{\Mem, \buf} \xrightarrow{\evt} \tup{\Mem,\buf}}
  \and
    \inferrule[read-from-memory]
  {\evt=\ev{\proc}{\RfM(\var,v)} \\\\ 
    \restrict{\buf(\proc)}{\var} = \emptyseq \\  \Mem(\var)=v}
  {\tup{\Mem, \buf} \xrightarrow{\evt} \tup{\Mem,\buf}}
   \\
  \inferrule[rmw]{\evt=\ev{\proc}{\RMW(\var,\vexp,\vnew)} \\\\ \buf(\proc) = \emptyseq \\ \Mem(\var)=\vexp}
{\tup{\Mem,\buf} \xrightarrow{\evt} \tup{\Mem[\var\mapsto \vnew],\buf}}
\and  
  \inferrule[fence]
  {\evt = \ev{\proc}{\Fence}
    \\\\ \buf(\proc) = \emptyseq }
  {\tup{\Mem, \buf} \xrightarrow{\evt} \tup{\Mem, \buf}} 
    \and
  \inferrule[propagate]
 {\evt = \ev{d_\proc}{\prop(x,v)}
    \\\\
  \buf(\proc) = \tup{\tup{\var,v}} \cdot \beta \\\\
  \Mem'= \Mem[\var \mapsto v] \\ b' = b[\proc \mapsto \beta]}
  {\tup{\Mem, \buf} \xrightarrow{\evt} \tup{\Mem', \buf'}} 
    \and
  \inferrule[null]
  {\evt = \ev{\proc}{\,\noop} \\
  }
  {\tup{\Mem, \buf} \xrightarrow{\evt} \tup{\Mem, \buf }} 
\end{mathpar}
Roughly speaking, actions operate as follows.
Read and write actions are always enabled. A write by process~$i$ places an item $\tup{var,v}$ at the tail of store buffer queue. A read $\R[x]$ by~$i$ will result in a read-from-buffer event $\RfB(x,v)$ if  $\tup{x,v}$ is the latest write to~$x$  in~$i$'s buffer; if the buffer contains no such pair, then a read-from-memory event $\RfM(x,v)$ takes place, returning the value $v=\Mem(x)$. Fence  and read-modify-write actions  by~$i$ block until their precondition is satisfied. In both cases the action requires~$i$'s buffer to be empty, and the action's success informs~$i$ that this was the case.
An $\RMW(\var,\vexp,\vnew)$ event occurs if, in addition, the current value of~$\var$ is~$\vexp$, resulting in the value of $x$ in memory being changed to~$\vnew$. Finally, a $\prop$ action by $d_i$ can succeed only if $i$'s buffer is nonempty, and it removes the oldest $\tup{x,v}$ pair from the buffer, and assigns the value~$v$ to~$x$. 
It will be useful to be able to distinguish actions that access memory from those that don't. For this purpose, we define: 

\begin{definition}[Memory access actions]
    A TSO action is considered a {\em memory access} of the variable $x$ if it is either an $\RMW[x,\cdot]$, it is a $\prop$ action that results in a $\prop(x,\cdot)$ event, or a read action~$\R[x]$ that results in a read from memory $\RfM(x,\cdot)$.   
\end{definition}
Notice that neither writes nor $\RfB$'s are considered memory accesses. This is because they directly interact only with the buffer store and not with the variables of shared memory. Roughly speaking, the occurrence of such an action cannot be observed by other processes before later writes by the process performing are  propagated to shared memory. 
\subsection{Runs and Protocols}
\label{sec:runs and protocols}
Our purpose is to study time and ordering in TSO systems, and to obtain insights into the structure of protocols that implementing various objects and tasks in this setting. A protocol is a tuple $P=(P_1,\ldots,P_n)$, containing an individual component $P_i$ for every $i\in\Pi$. Each of these is a nondeterministic function $P_i:L_i\to 2^{\Act_i}\setminus\emptyset$ 
where~$L_i$ is the set of local states of~$i$ and $\Act_i$ is its set of possible actions. $P_i(\ell)$ determines a nonempty set of possible actions that~$i$ can perform when in state~$\ell\in L_i$, out of which one action will be chosen nondeterministically whenever~$i$ moves in state~$\ell$. ($P_i$ is \emph{deterministic} if $P_i(\ell)$ is a singleton for all $\ell\in L_i$.) 

\paragraph{Actions:}\quad 
In addition to the TSO actions $\R$, $\W$, $\RMW$ and $\F$ for~$i$, it is convenient to allow in $\Act_i$ internal actions that will only be recorded in~$i$'s local state, and will not affect the TSO state. \footnote{ We consider a $\RfM$ in~$r$ and a $\RfB$ of the same tag in~$r'$ to be the ``same ($\R$) action''.}
One such action we will see in \Cref{sec:implOpTSO} is $\ret(\Op)$ which will mark the completion of an operation~$\Op$ invoked at~$i$. 
We will also allow the scheduler, which we sometimes call the environment, to perform $\invoke[i,\Op]$ actions, where, again,~$\Op$ is the name of an operation. (Operations will be discussed in  \Cref{sec:implOpTSO}.) 

\paragraph{States:}\quad 
The local state of a process records the history of its observations so far. 
Starting initially from $\ell_i=\lambda$, an empty list, whenever~$i$ performs a non-null action, its record is added to the list. 
For write, rmw and fence  actions the events $\W(x,v), \RMW(x,\vexp,\vnew)$ and~$\F$, respectively, will be recorded in the local state. A read action returning value~$v$ will be recorded as $\R(x,v)$, since~$i$ does not observe whether the read was from its buffer or from memory. Internal actions will similarly be recorded in the state, and if the environment performs $\invoke[i,\Op]$ then $\invoke(\Op)$ is added to~$\ell_i$. 
At any given point in time, we consider the system as being in a \emph{global state}, which is a tuple 
$(\sigma,\ell_1,\ldots,\ell_n)$ consisting of a TSO state $\sigma=\node{\Mem,\buf}$ and a local state for every process. 

\paragraph{Runs:}\quad 
We identify time with the natural numbers, and define a \emph{run} to be an infinite sequence of global states $r(0),r(1),r(2)\ldots$. In the initial state $r(0)$ all buffers in~$\sigma$ are empty ($\buf_i=\lambda)$, and all local states are initial local states. The transition from $r(t)$ to $r(t+1)$ is sometimes called \emph{round~$t+1$} of~$r$. 
For all times~$t\ge 0$, the global state $r(t+1)$ at time~$t+1$ is the result of an enabled \emph{joint action}
taking place at $r(t)$. This joint action $\vec{\act}=(\act_1,\ldots\act_n; \act_{d_1},\ldots,\act_{d_n};\act_e)$ 
contains an  action $\act_i\in\Act_i$ (possibly the null action~$\bot$) that is enabled at the state $r(t)$ for each $i\in\Pi$, an enabled action from $\{\prop,\bot\}$ for every local dispatcher~$d_i$, and an environment action $\act_e$ that may contain up to one $\invoke(i,\Op)$ message for each~$i\in\Pi$. (Operations~$\Op$ are discussed in greater detail in \Cref{sec:implOpTSO}.) We say that a joint action schedules an agent~$b$ to~{\em move} if $\act_b\ne\bot$. 
TSO stands for \emph{Total Store Order}, since it guarantees that memory accesses to any given variable must be totally ordered. Consequently, all actions in a given joint action must be non-conflicting, where two actions~$\act$ and $\act'$ 
are {\em conflicting} if (i) at least one of them is not a $\RfM$, and
(ii) they are both memory accesses to the same variable, and
(iii) it is not the case that one of them is $\prop$ at $d_i$ and the other is $\RfM$ at $i$.

\paragraph{Transitions by joint actions:}\quad 
Recall that in \Cref{sec:tsomodel} we defined the transition on the TSO state~$\sigma$ that is caused by a single TSO action. 
We define the transition  resulting from the joint action that $\vec{\act}=(\act_1,\ldots\act_n; \act_{d_1},\ldots,\act_{d_n};\act_e)$ being performed in $r(t)=(\sigma,\cdots)$, thereby producing $r(t+1)=(\sigma',\cdots)$ as follows. 
The state~$\sigma'$ is the result of starting from~$\sigma$ and performing the actions of $\vec{\act}$ one by one in the order they appear in~$\vec{\act}$. The local states are updated by adding the records of (non-null) actions to the processes according to~$\vec{\act}$.

\paragraph{Runs of a protocol:}\quad
We denote by $r_i(t)$ process~$i$'s local state in $r(t)$. 
A run~$r$ is considered a run of a protocol~$P=(P_1,\ldots,P_n)$ if for all $t\ge 0$ and $i\in\Pi$, 
if~$i$ performs a non-null action~$\act_i$ in the joint action~$\vec{\act}$ in round~$t+1$ of~$r$ at state $r(t)$, then $\act_i$ is 
consistent with~$P$. I.e., $\act_i\in P_i(r_i(t))$. 

\paragraph{Scheduling and Fairness:}\quad
None of our results before \Cref{sec:implOpTSO} depend on fairness assumptions, but protocols designed for TSO usually do. 
These involve the scheduling of processes and the propagation of messages from buffer stores. 
The first is that in every TSO run, every enabled $\prop$ action is eventually performed. Thus, if $\buf_i$ becomes nonempty at some~$t$ then $d_i$ will perform a $\prop$ action at a time $t'>t$. Notice that since a fence operation~$\F$ by~$i$ will  block until $\buf_i$ is empty, this assumption ensures that the fence will eventually be enabled. 
The second fairness condition is that a process that has an enabled action to perform, according to its protocol, will eventually move.

\paragraph{Nodes and Tags:}\quad 
Roughly speaking, we wish to state causal connections analogous to Lamport's happens-before between pairs of points along the timelines of (possibly different) processes. Given the subtle properties of TSO, we have added as distinct agents the local dispatchers~$d_i$ that are in charge of propagating messages from the store buffers. Considering points on their timelines as well as points on those of the processes will play an important role in allowing us to define the new relations. For an agent $b\in\Ag$ and a time~$t\ge 0$, we will define a \emph{node} $\theta=\node{b,t}$ that will serve to refer to the point at time~$t$ on~$b$'s timeline. With respect to a given run~$r$ and node~$\theta=\node{b,t}$, we will denote by $\theta.\alpha$ the action that agent~$b$ performs at time~$t$ in~$r$. (This can, in particular, be the null action~$\bot$.)

Another tool that will be used in defining our framework is the ability to relate values in the store buffer with the action that propagates them, and read operations that read them from the buffer or from memory. To this end, we define the auxiliary
notion of a tag, which consists of a process, and index, and a value as follows: 
\begin{definition}[$\alpha.tag$]
The $k$'th write event $\W(x,v)$ by a process~$i$ in a given run will have the field $\W.tag\triangleq\node{i,k}$.   Similarly, we associate a $tag$ field with $\prop(x,v)$, $\RfB(x,v)$ and $\RfM(x,v)$ events. Namely, we  have $\prop.tag\triangleq\node{i,k}$ (resp. $\RfB.tag\triangleq\node{i,k}$, $\RfM.tag\triangleq\node{i,k}$) if the corresponding action propagates (resp. reads) the item created by~$i$'s $k^{th}$ write operation in $r$. In the case of a read this will, of course, return only the value $v\in\Val$ written by this write. 
\end{definition}

\section{The Occurs-before Relation in TSO}
\label{sec:occurs-before}
In asynchronous message-passing systems, the only way that a protocol can ensure that an action at one process will take place later than a specific event that occurs at a second process, is by forcing the action to be delayed until a message chain from the second process has been constructed. This is Lamport's \emph{happens-before} relation (\cite{Lamclocks}). 
Our purpose now is to define a relation that similarly captures the timing information that may be available to processes in the TSO model. Notice that in the TSO model process scheduling and the timing of propagation is asynchronous. Consequently, processes do not have direct information about the timing and order of events. Further complicating matters is the fact that distinct processes can read different values of a variable~$x$ \emph{at the same time}. To mitigate this, TSO provides 
synchronization primitives such as fences and read-modify-write actions. 
Our relation, which we call \emph{occurs-before} will also involve creating a chain among points on the timelines of different processes.%
\footnote{We chose a new name for this relation, because ``happens before'' has become synonymous with message chains, and the chains created in the case of TSO look very different from message chains.} We  proceed as follows:
\begin{definition}[TSO occurs before]\label{def:obTSO}
    Let~$t<t'$, let $i,j\in\Pi$ and let $b,c\in \Ag$. We define the binary \emph{occurs before} relation $\ob_r$ between nodes in $r$  as follows:
    \begin{enumerate}
        \item\label{it:local} (locality) 
        $\node{b,t}\ob_r\node{b,t'}$ for every agent~$b$.
          \item (actions involving the buffer) \label{it:defhbprocprop} $\node{i,t}\ob_r\node{d_i,t'}$ if \label{it:defhbprocenv}
        a write at $\node{i,t}$ or an item that is read from $i$'s buffer at $\node{i,t}$ is propagated at $\node{d_i,t'}$, more formally: $\node{i,t}.\alpha\in\{\W,\RfB\}$,  $\node{d_i,t'}.\alpha=\prop$ and $\node{i,t}.\alpha.tag=\node{d_i,t'}.\alpha.tag$.

             \item \label{it:samevar} (memory accesses to a given variable) $\node{b,t}\ob_r\node{c,t'}$ if $\node{b,t}.\alpha$ and $\node{c,t'}.\alpha$ are both memory access actions of the same variable \textbf{\underline{except}} if 
             \begin{enumerate}
                 \item $\node{b,t}.\alpha$ and  $\node{c,t'}.\alpha$ are both $\RfM$ actions, or if 
                 \item \label{it:propRfM}  $b=d_i$, $c=i$,  $\node{d_i,t}.\alpha=\prop$ and $\node{i,t'}.\alpha=\RfM$.
             \end{enumerate}
        \item\label{it:fences} $\node{d_i,t}\ob_r\node{i,t'}$ if $\node{d_i,t}.\alpha=\prop$ and $\node{i,t'}.\alpha\in\{\F,\RMW\}$. 
        \item\label{it:transitive} (transitivity) 
        $\theta\ob_r\theta'$ if $\theta\ob_r\hat\theta$ and $\hat\theta\ob_r\theta'$ for some node $\hat\theta$.
        \end{enumerate}
        
\end{definition}

Notice that if $\theta\ob_r\theta'$ then these nodes must be connected by a chain  whose links are ``base steps'' obtained by clauses~\ref{it:local} to~\ref{it:fences}. The links are stringed into a chain by the transitivity clause~\ref{it:transitive}. 
Roughly speaking, the relation is intended to capture information that may be available to the processes regarding the ordering of events. Thus, for example,  exception (a) in clause~\ref{it:samevar} is motivated by the fact that if two processes read the same item from a given variable in memory, then the order among their reads will never be observable to the processes unless some event occurs in between the two events. Similarly, exception (b) is motivated by the fact that the information a process obtains when it reads a value that it has written does not by itself determine whether or not it has been propagated to memory by the time of the read. We remark that $\F$ and $\RMW$ actions are a signal that a process receives from its store buffer, implying that propagations have taken place. Clause~\ref{it:fences} accounts for this signal. 
Two immediate consequences of \Cref{def:obTSO} are:  

\begin{observation}\label{obs: ob occurs before}
If $\node{b,t}\ob_r\node{c,t'}$ then $t<t'$.
\end{observation}
To see why this is true, notice that all base cases of the definition (\ref{it:local}--\ref{it:fences}) relate a node at a time $t_1$ to a node at a time $t_2>t_1$, while the final clause (\ref{it:transitive}) transitively closes the relation. 
\Cref{obs: ob occurs before} justifies saying that if $\theta\ob_r\theta'$ then $\theta$ \emph{occurs before}~$\theta'$ in~$r$.%
\footnote{We have not called the relation \emph{happens before} because the latter is very closely associated with message chains, following~\cite{Lamclocks}, whereas $\ob_r$ is not.}
This observation implies that $\ob_r$ is irreflexive, and clause (\ref{it:transitive}) shows that it is transitive.
It follows that $\ob_r$ is a strict partial order relation on the nodes of~$r$. 

\begin{observation}\label{obs:si-iimpliesfence}
    If there is $\node{d_i,\cdot}\ob\node{i,t}$ obtained using a base case in \Cref{def:obTSO}, then $\node{i,t}.\alpha\in\{\F,\RMW\}$.
\end{observation}

What makes $\ob_r$ interesting and useful is the fact that, in a precise sense, it covers all the information that may be available to the processes regarding the ordering of events. Indeed, we will show that, roughly speaking, if an event~$e$ in~$r$ is not related by $\ob_r$ to another event~$e'$ in~$r$, then there is a run~$r'$ that is indistinguishable from~$r$ in which~$e'$ takes place strictly before~$e$ does. We now turn to prove a more general result which will imply this fact.
\section{Delaying the Future in TSO}
Our ultimate goal is to use our analysis, and specifically the occurs-before relation, to be able to prove lower bounds and necessary conditions on the structure of protocols that solve problems of interest. Roughly speaking, an important first step towards this will be to show that processes cannot be guaranteed that events are ordered in a particular fashion unless an occurs-before chain is formed between them. In general, if a process has the same local state at two points of different executions at which it is scheduled to move, then the protocol will specify the same action to be performed at both points. Consequently, events that are not reflected in the local state do not affect the behavior of the process. This motivates the following definition of $\Past_r(S)$, which in broad terms consist of all nodes that are in the causal past of the nodes of~$S$.

\begin{definition}[The past]
For a set of nodes $S$ in a run~$r$, we define 
    \[\Past_r(S)\triangleq\{\theta:\,\theta\ob_r\theta'\mbox{~for some }\theta'\in S\}\quad\mbox{and}\quad
    \Pastrp(S)~\triangleq~\Past_r(S)\,\cup\,S.\]
 
\end{definition}
See \Cref{fig:pasts} for an illustration.

Roughly speaking, the information available to a process at a given point is determined by its local state there. A process is unable to distinguish between runs in which it passes through the same sequence of local states. 
We will find it useful to consider when two runs cannot \emph{ever} be distinguished by any of the processes. 

\begin{definition}[Local Equivalence]
\label{def:loc-eq}
    Two runs $r$ and $r'$ are called \emph{locally equivalent}, denoted by  $r\loc r'$,  if for every process $j$, a local state $\ell_j$ of~$j$ appears in~$r$ iff $\ell_j$ appears in~$r'$. 
\end{definition}
Recall that the local state of a process~$i$ consists of its local history so far. Consequently, if two runs are locally equivalent, then every process starts in the same state, performs the same actions 
all in the same order, in both runs.

We will state and prove a slightly stronger property than discussed above. Namely, we will show that fixing a finite%
\footnote{All sets of nodes mentioned in the sequel are assumed to be finite.} 
set of nodes~$S\subseteq \Ag\times\mathbb{N}$ in a run~$r$, there is a locally equivalent run in which all actions that are performed in~$r$ at nodes that do not occur before nodes in~$S$ can be delayed by an arbitrary amount. In particular, this enables fixing the timing of nodes of~$S$ and moving such actions so that they all take place later in real time than all nodes in~$S$. 
Before we can state and prove this theorem we need to make a few technical definitions:

For every agent~$b\in\Ag$, we will be interested in identifying the first node on~$b$'s timeline that is not in the causal past:

\begin{figure}
    \centering
    \includegraphics[width=\linewidth]{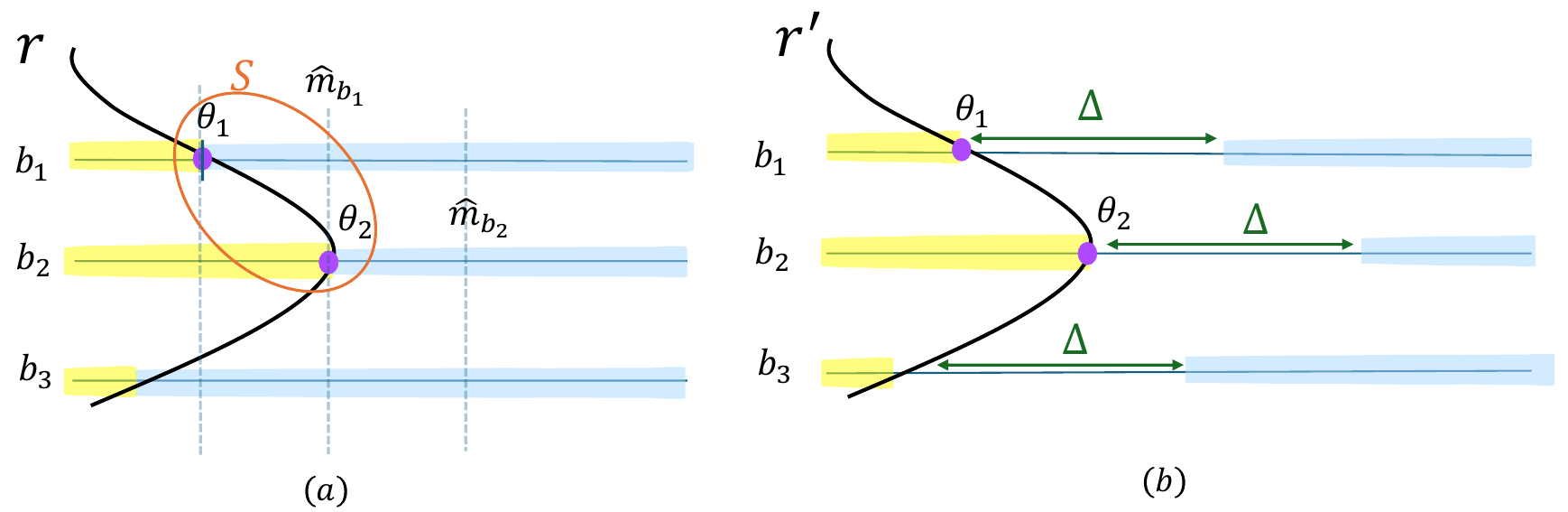}
    \caption{$S=\{\theta_1,\theta_2\}$ and $\Past_r(S)$ is composed of all the nodes in the yellow part. $\Past^+_r(S)$ contains in addition the nodes on the black curve, delimiting the nodes in past of $S$ from the other nodes. (b) represents the run $r'$ guaranteed to exist by \Cref{thm:DtFTSO}.}
    \label{fig:pasts}
\end{figure}
\begin{notation}
    Fix a set~$S$ of nodes and a run $r$. For each $b\in \Ag$ we define%
    \footnote{While the value of~$\hatm_b$ depends on~$S$, the set~$S$ is typically clear from context, so we write $\hatm_b$ rather than $\hatm^S_b$ for ease of exposition.} the time
    \item \begin{equation*}
       \hat{m}_b~\triangleq~
    	\begin{cases}    		
    		
    	0 & \node{b,0}\notin \Pastrp(S)\\
        t& t>0, ~\node{b,t-1}\in\Pastrp(S),~\text{and}~ \node{b,t}\notin \Pastrp(S) 
    	    	\end{cases}     
    \end{equation*}
   
\end{notation}

    We are now ready to state and prove a fundamental result concerning reordering of actions in TSO: 
\begin{restatable}[Delaying the future in TSO]{theorem}{primethm}\label{thm:DtFTSO}
   Let~$r$ be a TSO run~$r$ of a protocol~$P$, let $S$ be a set of nodes of~$r$, and let $\Delta\ge 0$. 
    Then there exists a TSO run~$r'\approx r$ of the same protocol~$P$ such that: 
    
\begin{itemize}
    \item[(a)] Every agent~$b\in\Ag$ performs exactly the same actions in the same order in both runs;
     moreover, an action that $b$ performs at time~$t$ in~$r$ is performed in~$r'$ at time~$t$ if $\node{b,t}\in\Past_r(S)$, and it is performed at time $t+\Delta$ otherwise.     
    \item[(b)] Every process~$j$'s local states are shifted accordingly:
    \begin{equation*}
    	{r}_j(t)=
    	\begin{cases}    		
    		{r'}_j(t)& \textrm{for all~}\node{j,t}\in\Past_r^+(S)\\
    	{r'}_j(t+\Delta) & \textrm{otherwise}.
        \end{cases}
         \end{equation*}
\end{itemize}
 \end{restatable}

\begin{proofsketch}
To show (a) and (b), we define $r'$ as follows. 
$r'(0)\triangleq r(0)$, and every agent~$b$ will act as follows: 
For times $t<\hatm_b$, its action~$\alpha$ at time~$t$ in~$r'$ will be the same as its at time~$t$ in~$r$ if this action is enabled and the agent's local state is the same at $r'(t)$, and will be $\noop$ otherwise; 
for $\hatm_b\le t<\hatm_b+\Delta$, agent~$b$ will perform~$\noop$ at time~$t$ in~$r'$; and finally for 
time $t\ge \hatm_b+\Delta$ agent~$b$ will perform in~$r'$ the action it performs at time~$t-\Delta$ in~$r$ if that is enabled at $r'(t)$, and $\noop$ if it is not enabled.  
The proof then proceeds to show by induction on~$k$ that 
(i) the local state $r'_i(k)$ of every process~$i$ at time~$k$  in~$r'$ is the same as its local state at the corresponding time ($r_i(k)$ or $r_i(k-\Delta)$ depending on whether $k<\hatm_i$), establishing 
 that part (b) holds up to time~$k$, and 
 (ii) that the action it performs at the corresponding time in~$r$, if any, is enabled at $r'(k)$; and similarly
 (iii) for every local dispatcher $d_i$, the buffer $\buf_i$ is nonempty and the oldest value each one contains are identical
 at $r'(k)$ if 
 $d_i$ performs $\prop$ at the corresponding time (depending on whether $k<\hatm_{d_i}$) in~$r$. 
It follows that (a) holds up to time~$k$ as well. Finally, we show using part (b) that $r'$ is locally equivalent to $r$.
A full proof with complete details appears in the Appendix.

One of the most subtle cases, for example, arises when a $\prop$ action by $d_i$ is shifted forward by $\Delta$ rounds in~$r'$ while a read from memory of the propagated tag by process~$i$ is not. Hence, by definition of~$r'$, the read action is performed in~$r$ at a node of $\Past^+_r(S)$, while the $\prop$ is performed at a node that is not in $\Past^+_r(S)$. If the shift is such that the $\prop$ occurs after the read in~$r'$, then the $\RfM$ of~$x$ in~$r$  will become an $\RfB$ in~$r'$.  
Clearly, the $\RfB$ must read a tag that~$i$ has written, so we need to show that the tag read by the $\RfM$ in~$r'$ was the same one written by~$i$. This is guaranteed by the definition of $\ob$. We proceed as follows. Suppose the read tags are not the same. Then a different tag is propagated  to~$x$ in~$r$ (by some $d_j\ne d_i$)  after the $\prop$ by $d_i$ and before~$i$ reads.  Since all three actions are memory accesses to~$x$, we have by  clause~3 of \Cref{def:obTSO} that $d_i$'s $\prop$ occurs before $d_j$'s $\prop$ in~$r$ and the latter occurs before~$i$'s $\RfM$. But then both $\prop$'s occur at nodes of $\Past^+_r(S)$, in~$r$, contradicting the assumption that $d_i$'s $\prop$  was shifted forward in time in~$r'$. It follows that~$i$ reads the same tag in both runs, as claimed.
\end{proofsketch}
\begin{toappendix}
To simplify the case analysis in the proof of \Cref{thm:DtFTSO}, we define:
\begin{definition}[The $\mathtt{shift}$ operator]
For $k,t,\Delta\ge 0$
we define 
\begin{equation*}
    	\shift{t,k} ~\triangleq~
    	\begin{cases}    		
    		t& t\leq k\\
    	t+\Delta & t\geq k+1
    	    	\end{cases}     
    \end{equation*}
\end{definition}
We now restate and prove \Cref{thm:DtFTSO}.

\primethm*

For the proof purposes we first introduce some definitions and useful claims.
The fact that the set $\Past_r(S)$ is backward closed under $\ob_r$ immediately implies that the shift operator as we will use it does not reorder nodes that stand in the $\ob_r$ relation:
\begin{claim}\label{claim:order}
            Fix~$S$, a run~$r$ and nodes $\theta_1=\node{b_1,t_1}$ and $\theta_2=\node{b_2,t_2}$. 
            If $\theta_1\ob_r\theta_2$
             then $\shift{t_1+1,\hatm_{b_1}}<\shift{t_2+1,\hatm_{b_2}}$.
          \end{claim}
        \begin{proof}
        By \Cref{obs: ob occurs before}
        we have $t_1<t_2$ (*) Observe in addition that for every agent $b$ and time $t$, if $\node{b,t}\in\Past^+(S)$ then $t+1\leq \hatm_b$.
        We now reason by cases:
        \begin{itemize}
            \item If $\theta_2\in\Past^+(S)$ then by definition, $\theta_1\in\Past^+(S)$. So, $\shift{t_1+1,\hatm_{b_1}}=t_1+1$ and $\shift{t_2+1,\hatm_{b_2}}=t_2+1$.
            \item If $\theta_2\notin\Past^+(S)$ and $\theta_1\notin\Past^+(S)$, then $\shift{t_1+1,\hatm_{b_1}}=t_1+1+\Delta$ and $\shift{t_2+1,\hatm_{b_2}}=t_2+1+\Delta$.
            \item If $\theta_2\notin\Past^+(S)$ and $\theta_1\in\Past^+(S)$, then $\shift{t_1+1,\hatm_{b_1}}=t_1+1$ and $\shift{t_2+1,\hatm_{b_2}}=t_2+1+\Delta$.
        \end{itemize}
        In every case, together with (*), we have that $\shift{t_1+1,\hatm_{b_1}}<\shift{t_2+1,\hatm_{b_2}}$.
        \end{proof}

\begin{claim}\label{claim:shift-1welldefined}
     $\shift{t+1,\hatm_b}-1\geq 0$
     for all $t\geq 0$ and for all agents $b$.
\end{claim}
 \begin{proof}
     If $t < \hatm_b$, then $\shift{t+1, \hatm_b} = t+1$, so $\shift{t+1, \hatm_b} - 1 = t \geq 0$; and if $t \ge \hatm_b$, then $\shift{t+1, \hatm_b} = t+1 + \Delta$, hence again $\shift{t+1, \hatm_b} - 1 \geq t \geq 0$.
 \end{proof}
    We will use in the proof the following useful notation:
    \begin{notation}
    We denote the contents of process~$i$'s buffer at time $t$ in run~$r$ by $\buf(i,r,t)$, and its elements by $(\buf(i,r,t)[0],\buf(i,r,t)[1],\dots)$. Similarly, $x(r,t)$ denotes the content of variable $x$ in the memory at time $t$ in $r$.
\end{notation}

    \begin{proof}
We shall construct a run ${r'}\approx r$ satisfying, for every process~$j$, agent $b$, all $m\ge 0$:
                     
                      \begin{enumerate}[label=(\arabic*)]
             
                   \item \label{it:i}  $r_j(m)=r'_j(\shift{m,\hatm_{j}})$ for all $m\ge 0$,  
                   \item \label{it:ii} $b$ performs exactly the same actions and reads the same tags in round~$m$ of~$r$ and in round $\shift{m,\hatm_b}$ of~${r'}$, for all $m\ge 1$.
        \end{enumerate}

We construct ${r'}$ as follows. Both runs start  in the same initial state: ${r'}(0)\triangleq r(0)$.
For every agent~$b$ and 
for all $\tilde{m}$ in the range 
$\hatm_b+1\le \tilde{m}\leq \hatm_b+\Delta$ agent $b$ does not move in round $\tilde{m}$ of $r'$. For all $m\ge 1$ and every local dispatcher~$d_i$, if $\node{d_i,m-1}.\alpha=\prop$ and 
$\buf(i,r',\shift{m,\hatm_{d_i}}-1)\neq\epsilon$ then $d_i$ performs a $\prop$ in round $\shift{m,\hatm_{d_i}}$ of $r'$.
Otherwise, $d_i$ does not move in round $\shift{m,\hatm_{d_i}}$ of $r'$. 
For all processes~$j$ and $m> 0$, if (i) $j$ moves in round $m$ of $r$, (ii) ${r'}_j(\shift{m,\hatm_{j}}-1)=r_j(m-1)$ and
(iii) the preconditions of $\node{j,m-1}.\alpha$ in $r$ hold in $r'$ at time $\shift{m,\hatm_{j}}-1$, then $j$ moves in round $\shift{m,\hatm_{j}}$ of $r'$ and $j$ performs the same action $\alpha_j\in P_j(r_j(m-1))$ in round $\shift{m,\hatm_{j}}$ of~${r'}$ as in round~$m$ of~$r$, and otherwise $j$ does not move.
Observe that, by definition, all processes follow the protocol~$P=(P_1,\ldots,P_n)$ in~${r'}$.
\Cref{claim:shift-1welldefined} ensures that all references to time points of the form $\shift{t+1, \hatm_b} - 1$ in our constructions are well-defined and non-negative.

Moreover, observe the following useful property of~${r'}$:
 
\begin{claim}\label{claim:m-1}
 For all $m>0$ and for every process $j$: ${r'}_j(\shift{m,\hatm_{j}}-1)={r'}_j(\shift{m-1,\hatm_{j}})$.
\end{claim}
\begin{proof}
      We consider two cases: 
\begin{itemize}
\item $m=\hatm_{j}+1$:\quad Recall that, by definition of~$r'$, we have that $j$ does not move in rounds $\hatm_{j}+1\leq m'\leq \hatm_{j}+\Delta$. Consequently, 
${r'}_j(\hatm_{j}+\Delta)={r'}_j(\hatm_{j}+\Delta-1)=\dots={r'}_j(\hatm_{j})$. 
Hence, ${r'}_j(\shift{m,\hatm_{j}}-1)={r'}_j(\shift{\hatm_{j}+1,\hatm_{j}}-1)={r'}_j(\hatm_{j}+1+\Delta-1)={r'}_j(\hatm_{j}+\Delta)={r'}_j(\hatm_{j})={r'}_j(\shift{m-1,\hatm_{j}})$.
    \item $0<m\ne \hatm_{j}+1$: 
    If $m\leq \hatm_{j}$ then by definition of  $\mathtt{shift}_\Delta$ we have  that  \mbox{$\shift{m,\hatm_{j}}=m$} and $\shift{m-1,\hatm_{j}}=m-1=\shift{m,\hatm_{j}}-1$. 
    Similarly, if $m>\hatm_{j}+1$ then $\shift{m,\hatm_{j}}=m+\Delta$ and $\shift{m-1,\hatm_{j}}=m-1+\Delta=\shift{m,\hatm_{j}}-1$. In both cases we obtain that ${r'}_j(\shift{m,\hatm_{j}}-1)={r'}_j(\shift{m-1,\hatm_{j}})$, as desired. \vspace{-3mm}
\end{itemize}
\end{proof}

 We are now ready to prove that~${r'}$ satisfies \ref{it:i} for all processes $j$ and \ref{it:ii} for all agents~$b$. We prove this by induction on~$m\ge 0$.

 \uline{Base:} $m=0$. Observe that $\shift{0,\hatm_{j}}=0$ for all $j$.
 By definition of ${r'}$ we have that $r'_j(0)=r_j(0)$, ensuring \ref{it:i}. 
 \ref{it:ii} holds vacuously because there is no round $m=0$.
 
 \uline{Step:} Let $m>0$ and assume inductively that \ref{it:i} and \ref{it:ii} hold for all agents at all times strictly smaller than~$m$.
We now show that each possible TSO action occurs at $\theta$ in $r$ iff it occurs at the corresponding node of $\theta$ in $r'$.

\textbf{The following claim establishes that propagate actions for a given tag occur at corresponding nodes in both the original run~$r$ and~$r'$.}
\begin{claim}\label{claim:propmem}
For every process $i$, tag $\kappa$ written by $i$ is propagated to variable $x$ in the memory in round $m$ of $r$, (i.e., $\node{d_i,m-1}.\alpha=\prop(x,\kappa)$) iff it is propagated to $x$ in the memory in round $\shift{m,\hatm_{d_i}}$ of~$r'$.
\end{claim}
\begin{proof}

We show both directions:
\begin{itemize}
    \item Let $\kappa$ be the tag propagated in round $m$ of $r$.  We need to show that in $r'$ at time $\shift{m,\hatm_{d_i}}-1$, tag $\kappa$ is at the the head of $i$'s buffer. Denote by $\tilde{m}$ the round at which $\kappa$ has been written by $i$ in $r$. By clause \ref{it:ii} of the inductive assumption we have that $\kappa$ is written by $i$ in $r'$ in round $\shift{\tilde{m},\hatm_{i}}$.
    By \Cref{def:obTSO}, $\node{i,\tilde{m}-1}\ob_r\node{d_i,m-1}$ so by \Cref{claim:order}, $\shift{\tilde{m},\hatm_{i}}<\shift{m,\hatm_{d_i}}$.
    By the induction assumption we have that all tags that have been propagated at $d_i$ before round $m$ in $r$ have been propagated before round $\shift{m,\hatm_{d_i}}$ in $r'$. We now show that $\kappa$ has not been propagated before round $\shift{m,\hatm_{d_i}}$ in $r'$.
    Assume by way of contradiction that there exists $m'>0$ such that $\shift{\tilde{m},\hatm_{i}}<\shift{m',\hatm_{d_i}}<\shift{m,\hatm_{d_i}}$ and $\kappa$ is propagated in $r'$ in round $\shift{m',\hatm_{d_i}}$. Then we obtain by clause  \ref{it:ii} of the inductive assumption and \Cref{claim:order} that $\kappa$ is propagated in round $m'<m$ of $r$. 
    This contradicts the fact that $\kappa$ has not been propagated in $r$ at time $m-1$. We thus conclude that $\kappa$ has been written by $i$ but not been propagated in round $\shift{m,\hatm_{d_i}}$ of $r'$. So, $\kappa$ is propagated as well in round $\shift{m,\hatm_{d_i}}$ of~$r'$.
\item The other direction is immediate: Assume $\kappa$ written by $i$ is propagated to $x$ in the memory in round $\shift{m,\hatm_{d_i}}$ of $r'$.
            By definition of $r'$, $\node{d_i,m-1}.\alpha=\prop$ in $r$ and the same preconditions hold at $(r,m-1)$ and $(r',\shift{m-1,\hatm_{d_i}})$, so
            $\kappa$ is propagated to the memory in round $m$ of $r$.

\end{itemize}

\end{proof}
\textbf{The following claim establishes that if a $\RMW$ or a $\F$ occurs at node $\theta$ of $j$ in $r$, then the buffer of $j$ is empty at the corresponding node of $\theta$ in $r'$.}
\begin{claim}\label{claim:RMWeps}

    If $\node{j,m-1}.\alpha\in\{\RMW,\F\}$ in $r$, then $\buf(r',\shift{m,\hatm_{j}}-1)=\epsilon$.
\end{claim}
\begin{proof}
    By their definitions, the $\RMW$ and $\F$ operations can take place in round~$m$ of~$r$ only if  $\buf(r,m-1)=\epsilon$. If there is no write action by~$j$ up to round $m$ in $r$, then by clause \ref{it:ii} of the inductive assumption, there is no write action by~$j$ up to round $\shift{m,\hatm_{j}}$ in $r'$ and therefore, $\buf(r',\shift{m,\hatm_{j}}-1)=\epsilon$. 
    Otherwise, let $\shift{\tilde{m},\hatm_{j}}<\shift{m,\hatm_{j}}-1$ such that $\node{j,\shift{\tilde{m},\hatm_{j}}}.\alpha=\W$ in $r'$ and denote by $\kappa$ the sequence number of this Write \raissa{action}. By clause \ref{it:ii} of the inductive assumption, we have that $\node{j,\tilde{m}}.\alpha=\W$ in $r$ and $\node{j,\tilde{m}}.\alpha.tag=\kappa$.
    Since $\buf(r,m-1)=\epsilon$ we have that there is $\tilde{m}<m_1<m-1$ such that $\node{d_j,m_1}.\alpha=\prop$ in $r$ and $\node{d_j,m_1}.\alpha.tag=\kappa$. By \Cref{def:obTSO} we obtain $\node{j,\tilde{m}}\ob_r\node{d_j,m_1}\ob_r\node{j,m-1}$.
    By \Cref{claim:order}, we have that $\shift{\tilde{m}+1,\hatm_{j}}<\shift{m_1+1,\hatm_{d_j}}<\shift{m,\hatm_{j}}$.
    By the inductive assumption \ref{it:ii}, $\kappa$ is propagated in round $\shift{m_1,\hatm_{d_j}}+1$ of~$r'$.
    I.e., $\kappa$ is propagated in some round earlier than $\shift{m,\hatm_{j}}$ in $r'$. This holds for every write operation by $j$. So, $\buf(r',\shift{m,\hatm_{j}}-1)=\epsilon$.
\end{proof}

\textbf{This implies that fences $\F$ occur at corresponding nodes in $r$ and $r'$.}
\begin{claim}\label{claim:F}
$\node{j,m-1}.\alpha=\F$ iff $\node{j,\shift{m,\hatm_{j}}-1}.\alpha=\F$.
\end{claim}
\begin{proof}
We show both directions.
\begin{itemize}

    \item Assume $\node{j,m-1}.\alpha=\F$.  We have by \Cref{claim:RMWeps} that $\buf(r',\shift{m,\hatm_{j}}-1)=\epsilon$ and by clause \ref{it:i} of the inductive assumption $r_j(m-1)=r'_j(\shift{m,\hatm_{j}}-1)$. So, by the definition of $r'$ we have $\node{j,\shift{m,\hatm_{j}}-1}.\alpha=\F$.
   \item The other direction is immediate. By the definition of $r'$, $\node{j,\shift{m,\hatm_{j}}-1}.\alpha=\F$ only if $\node{j,m-1}.\alpha=\F$.
\end{itemize}
   
\end{proof}

\textbf{The following claim establishes that if a $\RMW(x,\cdot,\cdot)$ occurs at node $\theta$ of $j$ in $r$, then the variable $x$ in the memory contains identical values at the corresponding node of $\theta$ in $r'$.}

\begin{claim}\label{claim:RMWv1}

    If $\node{j,m-1}.\alpha=\RMW(x,v_1,v_2)$ in $r$, then $x(r',\shift{m,\hatm_{j}}-1)=v_1$.
\end{claim}
\begin{proof}
    By the semantics of the $\RMW$ action, the fact that the action is enabled at time~$m-1$ in~$r$ implies that $x(r,m-1)=v_1$. We show that if a value $v'\neq v_1$ is propagated to $x$ in the memory in some round earlier than $m$, then there is a later round earlier than~$m$ at which $v_1$ is propagated to~$x$. Let $\shift{\tilde{m},\hatm_{d_i}}<\shift{m-1,\hatm_{j}}$ such that $\node{d_i,\shift{\tilde{m},\hatm_{d_i}}}.\alpha=\prop(x,v')$ with $v'\neq v_1$ in $r'$. By clause \ref{it:ii} of the inductive assumption, we have that $\node{d_i,\tilde{m}}.\alpha=\prop(x,v')$.
    Given that there is no conflicting actions in a joint action in $r$ and that $x(r,m-1)=v_1$, we have that there is 
    $\tilde{m}< m_1< m-1$ such that
    $\node{d_q,m_1-1}.\alpha=\prop(x,v_1)$ in~$r$ for some agent~$d_q$. By \Cref{def:obTSO} we get $\node{d_i,\tilde{m}}\ob_r\node{d_q,m_1-1}\ob_r\node{j,m-1}$.
    \Cref{claim:order} implies $\shift{\tilde{m}+1,\hatm_{d_i}}<\shift{m_1,\hatm_{d_q}}<\shift{m,\hatm_{j}}$.
    By clause \ref{it:ii} of the inductive assumption we obtain that $\node{d_q,\shift{m_1-1,\hatm_{d_q}}}.\alpha=\prop(x,v_1)$ in $r'$. So, $x(r',\shift{m,\hatm_{j}}-1)=v_1$.
\end{proof}

\Cref{claim:RMWv1} and \Cref{claim:RMWeps} imply that:
\begin{claim}\label{claim:RMW}
 $\node{j,m-1}.\alpha=\RMW(x,v_1,v_2)$ iff $\node{j,\shift{m,\hatm_{j}}-1}.\alpha=\RMW(x,v_1,v_2)$.
\end{claim}
\begin{proof}
We show both directions.
\begin{itemize}

    \item Assume $\node{j,m-1}.\alpha=\RMW(x,v_1,v_2)$. By \Cref{claim:RMWv1} we have that $x(r',\shift{m,m_j}-1)=v_1$ and by \Cref{claim:RMWeps} we have that $\buf(r',\shift{m,m_j}-1)=\epsilon)$.
So, by the definition of $r'$ we have that $\node{j,\shift{m,\hatm_{j}}-1}.\alpha=\RMW(x,v_1,v_2)$.
    \item The other direction is immediate. By the definition of $r'$, $\node{j,\shift{m,\hatm_{j}}-1}=\RMW(x,v_1,v_2)$ only if $\node{j,m-1}.\alpha=\RMW(x,v_1,v_2)$.
\end{itemize}
\end{proof}

\textbf{The following claim implies that a tag $\kappa$ is read in $r$ iff it is read at the corresponding node in $r'$.}
\begin{claim}\label{claim:reads}
   Process $j$ reads a value with tag $\kappa$ from variable $x$ in round $m$ of $r$ iff $j$ reads the same value with tag~$\kappa$ from $x$ in round $\shift{m,\hatm_{j}}$ of $r'$. 
\end{claim}
\begin{proof}
Denote by $\kappa$ the tag read by $j$ at $(r,m-1)$. Clearly, $\kappa$ has been written in $r$ by some process $q$ before round $m$. By clause \ref{it:ii} of the inductive assumption, $\kappa$ is written in $r'$ before round 
$\shift{m,m_q^+}$. 
In addition, by the construction of $r'$ and by clause \ref{it:i} of the inductive assumption we have that $\node{j,\shift{m,\hatm_{j}}-1}.\alpha$ in $r'$ is a read operation.
We distinguish between two cases whether the read in $r$ is from buffer or from memory. I.e., if $\node{j,m-1}.\alpha=\RfM$ or $\node{j,m-1}.\alpha=\RfB$.
\begin{itemize}
    \item If $\node{j,m-1}.\alpha=\RfB$.
    Assume by way of contradiction that $\kappa'\neq\kappa$ is read by~$j$ at $(r',\shift{m,\hatm_{j}}-1)$. Observe that $\kappa'$ could not have been written by $j$, as this would immediately contradict (by clause \ref{it:ii} of the inductive assumption) the fact that $j$ reads $\kappa$ at $(r,m-1)$. 
    Thus, $\kappa'$ is read by $j$ at $(r',\shift{m,\hatm_{j}}-1)$ from the memory. 
    Observe in addition that since $\kappa$ is written by $j$, we have that $\kappa$ has been propagated to the memory in $r'$ before round $\shift{m,\hatm_j}$. Let $\tilde{m}>0$ such that $\kappa$ is propagated in $r'$ in round $\shift{\tilde{m},\hatm_{d_j}}$.
    Clearly, $\shift{\tilde{m},\hatm_{d_j}}<\shift{m,\hatm_{j}}$. By the definition of $r'$, $\kappa$ is propagated in round $\tilde{m}$ of $r$.
    We now show that $\tilde{m}<m$, reaching a contradiction to the fact that $\kappa$ is read from buffer in round $m$ of~$r$.
    Assume $m<\tilde{m}$. Then, by \Cref{def:obTSO}(\ref{it:defhbprocprop}), we have that $\node{j,m-1}\ob_r\node{d_j,\tilde{m}-1}$. By \Cref{claim:order} we obtain $\shift{m,\hatm_{j}}<\shift{\tilde{m},\hatm_{d_j}}$. This contradicts the fact that $\shift{\tilde{m},\hatm_{d_j}}<\shift{m,\hatm_{j}}$.
    
    \item  If $\node{j,m-1}.\alpha=\RfM$. Recall that $\kappa$ is the tag read by $j$ at $(r,m-1)$. We denote by $\node{d_q,\tilde{m}-1}$ the node at which $\kappa$ is propagated in $r$. Assume by way of contradiction that $\kappa'\neq\kappa$ is returned by $j$ at $(r',\shift{m,\hatm_{j}}-1)$.
\begin{itemize}
    \item If $\kappa'$ is written by a process $i\neq j$, then there exists $m'>0$ such that~$\kappa'$ is propagated at $(r',\shift{m',\hatm_{d_i}}-1)$  with $\shift{\tilde{m},\hatm_{d_q}}<\shift{m',\hatm_{d_i}}<\shift{m,\hatm_{j}}$. 
    By the definition of $r'$, we have that $\kappa'$ is propagated to $x$ in round $m'$ of~$r$. By the definition of $r'$ and since there is no conflicting actions in joint actions of~$r$, we have that $\tilde{m}\neq m'$ and $m'\neq m$.
        We show that $\tilde{m}<m'<m$. Assume this is not the case.
        \begin{itemize}
            \item If $m'<\tilde{m}$, then by \Cref{def:obTSO}(\ref{it:samevar}) $\node{d_i,m'-1}\ob_r\node{d_q,\tilde{m}-1}$. It follows by \Cref{claim:order} that $\shift{m',\hatm_{d_i}}<\shift{\tilde{m},\hatm_{d_q}}$. This contradicts the fact that $\shift{\tilde{m},\hatm_{d_q}}<\shift{m',\hatm_{d_i}}$.
            \item If $m'>m$, then by \Cref{def:obTSO}(\ref{it:samevar}), $\node{j,m-1}\ob_r\node{d_i,m'-1}$. It follows by \Cref{claim:order} that $\shift{m,\hatm_{j}}<\shift{m',\hatm_{d_i}}$, This contradicts the fact that $\shift{m',\hatm_{d_i}}<\shift{m,\hatm_{j}}$.
        \end{itemize}
        So, $\tilde{m}<m'<m$, contradicting that $\kappa$ is read by $j$ at $(r,m-1)$.
        \item If $\kappa'$ is written by process $j$. Let $m'>0$ such that~$\kappa'$ is written by~$j$ in round $\shift{m',\hatm_{j}}<\shift{m,\hatm_{j}}$ of~$r'$.
        \begin{itemize}
            \item  If $\kappa$ is written by $j$ (i.e., $j=q$), then by the definition of $\prop$ and $\R$ we have that $\shift{\tilde{m},\hatm_{j}}<\shift{m',\hatm_{j}}$.
        By the construction of $r'$, we have that $\kappa'$ is written by $j$ in round $m'$ of $r$.
        By the locality clause of \Cref{def:obTSO} and \Cref{claim:order}, we have $\tilde{m}<m'<m$. By the definition of $\prop$ and $\R$ in TSO, we have that $\kappa$ is not returned in the read of $\node{j,m-1}$ in $r$, thus having a contradiction.
        \item If $q\neq j$, i.e., $\kappa$ is not written by $j$.
        \begin{itemize}
            \item If $\kappa'$ is read from memory, then similarly to the previous case where $\kappa'$ is written by $i\neq j$, we reach a contradiction.
            \item If $\kappa'$ is read from buffer, i.e., it is not propagated before round $\shift{m,\hatm_{j}}$ of $r'$. Then $\kappa'$ is not propagated before round $m$ of $r$ (otherwise we would obtain a contradiction to clause \ref{it:ii} of the inductive assumption). This contradicts the fact that $\kappa$ is read by $j$ at $\node{j,m-1}$.
        \end{itemize}
        \end{itemize}
       
\end{itemize}
\end{itemize}


We showed that if $\kappa$ is read at $\node{j,m-1}$ in $r$, then it is read at $\node{j,\shift{m,\hatm_{j}}-1}$ in $r'$. 
The other direction is now immediate. Assume $\kappa$ is read at $\node{j,\shift{m,\hatm_{j}}-1}$ in $r'$. By definition of $r'$ there is a read operation at $\node{j,m-1}$ in $r$. Denote by $\kappa'$ the tag being read there. We showed that $\kappa'=\kappa$.
To conclude, $j$ reads tag $\kappa$ from variable $x$ in round $m$ of $r$ iff $j$ reads $\kappa$ from~$x$ in round $\shift{m,\hatm_{j}}$ of $r'$.
\end{proof}
Recall that we have by clause \ref{it:i} of the inductive assumption that $r'_j(\shift{m-1,\hatm_{j}})=r_j(m-1)$. \Cref{claim:m-1} thus implies that
\begin{equation}\label{eq:induc}
   ~~~ r'_j(\shift{m,\hatm_{j}}-1)~=~r_j(m-1).
\end{equation}

We can now show that \ref{it:i}-\ref{it:ii} hold for every process~$j$, agent $b$ and~$m$ by cases depending on whether $j$ and $b$ move or not in round $m$ of $r$.
\begin{itemize}
    \item If $j$ does not move in round $m$ of $r$:\quad By definition of ${r'}$, we have that $j$ does not move in round $\shift{m, \hatm_{j}}$ of $r'$, so clearly \ref{it:i} holds  and in addition, ${r'}_j(\shift{m,\hatm_{j}})={r'}_j(\shift{m,\hatm_{j}}-1)=r_j(m-1)=r_j(m)$, proving \ref{it:i}.
    
    \item If $j$ moves in round $m$ of $r$:\quad We showed that the preconditions of the action performed in $(r,m)$ hold at $(r',\shift{m,\hatm_{j}})$. Thus,  by the definition of ${r'}$, process $j$ moves in round $\shift{m,\hatm_j}$ and performs the same action $\alpha_j\in P_j(r_j(m))$ in the round $\shift{m,\hatm_{j}}$ of ${r'}$ as it does in the round $m$ of $r$. If $\alpha_j$ contains a read ($\R$ or $\RMW$) then by \Cref{claim:reads} and \Cref{claim:RMW}, we have that the same value is read by $j$ in the round $m$ of $r$ and round $\shift{m,\hatm_{j}}$ of $r'$, ensuring \ref{it:ii} holds. If $\alpha_j=\W$ then $j$ performs the same write in round $\shift{m,\hatm_{j}}$ of $r'$, ensuring the same values are appended to the buffer in the same order in $r$ and $r'$. Thus, identical actions are performed in round $m$ of $r$ and round $\shift{m,\hatm_{j}}$ of $r'$ and same values are returned, ensuring \ref{it:i} and \ref{it:ii}.
    \item Similarly to previous item, if dispatcher $b$ does not move in round $m$ of $r$ then it does not in round $\shift{m,\hatm_b}$ of $r'$. If it does, we have by \Cref{claim:propmem} that exactly the same action occurs in round $\shift{m,\hatm_b}$ of $r'$ as in in round $m$ of $r$.
    \end{itemize}

\end{proof}
 Observe that a consequence of \ref{it:i} and \ref{it:ii} is that  for all nodes $\theta=\node{b,t}$ and $\theta'=\node{c,t'}$ it is the case that 
 \begin{equation}\label{eq:ob_preserving}
     \node{b,t}\ob_r\node{c,t'} ~~\text{iff}~~ \node{b,\shift{t,\hatm_b}}\ob_{r'}\node{c,\shift{t',\hatm_c}}.
 \end{equation}

A useful immediate consequence of \Cref{def:obTSO}(\ref{it:defhbprocenv}) and \ref{def:obTSO}(\ref{it:samevar})  is:
\begin{observation}\label{obs:conflictimpliesob}
    If conflicting actions are performed at  nodes $\theta=\node{b,t}$ and $\theta'=\node{b',t'}$ in~$r$ with $t<t'$, then  $\theta\ob_r\theta'$.
\end{observation}

To complete the proof that $r'$ is a TSO run, we show:
    \begin{claim}
        Joint actions in $r'$ do not contain conflicting actions.
    \end{claim}
\begin{proof}
This follows from \Cref{eq:ob_preserving}.
    Let $m_1,m_2>0$ such that conflicting actions are performed at the nodes $\theta\triangleq\node{b,\shift{m_1+1,\hatm_b}-1}$ and $\theta'\triangleq\node{c,\shift{m_2+1,\hatm_c}-1}$ in $r'$ and such that $\shift{m_1+1,\hatm_b}\leq \shift{m_2+1,\hatm_c}$.  By \Cref{obs:conflictimpliesob} we have that $\theta\ob_{r'}\theta'$. By \Cref{eq:ob_preserving}, we have that $\node{b,m_1}\ob_r\node{c,t_2}$. By \Cref{claim:order}, $\shift{m_1+1,\hatm_b}<\shift{m_2+1,\hatm_c}$. Hence, $\theta$ and $\theta'$ do not occur at the same time. 

    
\end{proof}

\end{toappendix}

\Cref{thm:DtFTSO} is inspired by the ``Delaying the Future'' theorem of Moses and Nataf \cite{CommRequirDISC24}, which shows a (slightly weaker, but essentially the same) property relating ordering and the happens-before relation in asynchronous message passing systems. Our result exposes a new connection between the causal structure of classic asynchronous systems and TSO. 

\Cref{thm:DtFTSO} provides us with a concrete handle on coordination and timing in TSO: In a precise sense, two events are guaranteed to occur in a given temporal order if and only if the execution contains an occurs-before chain between them. Indeed, this also applies to complex operations. Intuitively, whenever the results of an operation depend on the operations that preceded it, the occurs-before relation must be established. 
The theorem implies that in the absence of an $\ob$ connection, operations that precede a given operation can be shifted so that they no longer precede it, without the results of the operation of interest changing. Moreover, we can show that, at least when a process runs solo  (or when it cannot detect that it has not run solo) while performing an operation, the only way to ensure an $\ob$ chain from the operation to future operations is by performing synchronization operations. We provide further details in the next section.

\section{Implementing Operations in TSO}\label{sec:implOpTSO}
Distributed objects in a model such as TSO are typically implemented using a set of methods (we will, rather, use the term operations) whose specification typically draws a causal connection between them. Thus, for example, an implementation of a linearizable register must satisfy strict relationship between a read operation and the writes that precede it in real time. Using \Cref{thm:DtFTSO}, it is possible to show that occurs-before chains must be constructed by any implementations of such objects. 
This requirement can then be showed to obtain insights into the structure of such implementations which, in particular, enables us to identify cases in which synchronization actions must be performed.

To leverage the strength of \Cref{thm:DtFTSO}, we now aim to demonstrate how operations on distributed objects can be reordered while preserving local equivalence. We consider operations that are associated with individual processes. An operation $\Op$
starts with an invocation input $\inv(i,\Op)$ from the environment to process~$i$, and ends when process~$i$ performs a matching response action   $\ret(\Op)$ action. 
Operation invocations in our model are nondeterministic and asynchronous --- they can be invoked at arbitrary times.\footnote{We assume for simplicity that following an $\inv(i,\cdot)$, the environment will not issue another $\inv(i,\cdot)$  before~$i$ has issued a matching $\ret$ completing the first operation.}  For the purpose of our analysis in this section, we make a few fairly standard and straightforward definition.

We say that an operation $\X$ occurs between nodes $\theta=\node{i,t}$ and $\theta'=\node{i,t'}$ in $r$ if an $\inv(i,\X)$ action by the environment occurs in round~$t+1$ in $r$ and process $i$ performs the matching $\ret(\X)$ action in round~$t'$. In this case we denote $\X.s\triangleq\theta$ and 
$\X.e\triangleq\theta'\!$, and use $t_{\X.s}(r)$ to denote the operation's starting time~$t$ and $t_{\X.e}(r)$ to denote its ending time~$t'$. We do not mention the run~$r$ in these notations when it is clear from the context.
An operation~$\Op$  {\em completes} in a run~$r$ if $r$ contains both its invocation and its response.

           \begin{definition}[Real-time order and Concurrency]
   For two operations $\X$ and $\Y$  in $r$ we say that $\X$ {\em precedes} $\Y$ in~$r$, denoted $\X<_{r}\Y$,  if $t_{\X.e}(r)<t_{\Y.s}(r)$, i.e., if~$\X$ completes before~$\Y$ is invoked. 
        If neither $\X<_{r}\Y$ nor $\Y<_{r}\X$, then $\X$ and $\Y$ are considered {\em concurrent} in~$r$. 
\end{definition}
 
We will use the following terminology in our analysis of operations: 
\begin{definition} Let~$\X$ and~$\Y$ be operations. 
\begin{itemize}
   \item  We write $\X\ob_r\Y$ and say that $\X$ \emph{occurs before}~$\Y$ in~$r$ if $\X.s\ob_r\Y.e$. (This generalizes~$\ob$ to operations.) 
   \item We say that process $i$ \emph{runs solo} during an operation~$X$ in~$r$ if the only agents taking actions in~$r$ between $t_{X.s}$ and $t_{X.e}$ are~$i$ and $d_i$. We say that an operation $\X$ performed by $i$ runs solo if $i$ runs solo while performing $\X$.
    \item $\X$ is said to {\em run in isolation in}~$r$ if no operation is concurrent to $\X$ in $r$.
\end{itemize}
\end{definition}

\subsection{The Need for Synchronization for Establishing Occurs Before}\label{sec:applyDtF}
Suppose that a distributed object employs operations $\X$ and~$\Y$ for which we can show that $\X<_r\Y$ requires that $\X\ob_r\Y$. This is very common, for example, in the case of linearizable objects; we will demonstrate examples in \Cref{sec:linearizable}. Then the structure of occurs-before chains will imply that costly synchronization actions are sometimes needed. 
We will show a more specific result, stating that in the absence of communication, the high cost must be paid. 
In this section we establish a connection between running solo, feedback loops, and the need for fences or rmw operations. 
This connection will be used in \Cref{sec:linearizable} to establish synchronization costs for concrete objects. 
We start with a slightly technical lemma concerning occurs-before chains that use at most two processes. 
   
   \begin{definition}[$\ijj$-only chain]
       Let $\theta_1=\node{i,t_1}$, $\theta_2=\node{j,t_2}$ such that $\theta_1\ob_r\theta_2$. If for all $\node{b,\cdot}$ in the occurs-before chain, it is the case that $b\in\{i,d_i,j,d_j\}$, we say that the chain is $\ijj${\em -only.}
   \end{definition}
   
We can show:

\begin{lemma}\label{lem:obcomposedOnlyOf_i_And_j}
        Let $\node{i,t_1}\ob_r\node{j,t_2}$ with $i\ne j$ be an $\ijj$-only chain. Then there must be a time~$t$ satisfying
        $t_1\leq t<t_2$ such that one of the following conditions holds:
        \begin{enumerate}
            \item \label{lem:obcomposedOnlyOf_i_And_jit1} $\node{i,t}.\alpha=\RMW$ and there exists $t<t'<t_2$ such that $\node{j,t'}.\alpha\in \{\F,\RMW,\RfM\}$, or
            \item \label{lem:obcomposedOnlyOf_i_And_jit2} $\node{i,t}.\alpha\in\{\W,\RfB\}$ and there exist $t<t'<t''<t_2$ such that $\node{d_i,t'}.\alpha=\prop$ and $\node{j,t''}.\alpha\in\{\RfM,\F,\RMW\}$, or
            \item \label{lem:obcomposedOnlyOf_i_And_jit3} $\node{i,t}.\alpha=\RfM$ and there exists $t<t'<t_2$ such that $\node{j,t'}.\alpha\in\{\F,\RMW\}$
        \end{enumerate}
       
\end{lemma}

\begin{proof}
    Let $\pi$ be such a chain. Observe first that since $i\neq j$, the chain does not contain only links obtained by applying the locality clause. Thus, the chain must contain:
    \begin{itemize}
        \item a base link of the form $\node{i,\cdot}\ob_r\node{j,\cdot}$. This can be obtained only if \Cref{lem:obcomposedOnlyOf_i_And_jit1} holds or \Cref{lem:obcomposedOnlyOf_i_And_jit3} holds, or
        \item a base link of the form $\node{i,\cdot}\ob_r\node{d_i,\cdot}$. This can be obtained only if \Cref{lem:obcomposedOnlyOf_i_And_jit2} holds.
    \end{itemize}
\end{proof}    

    Information in TSO can be conveyed by receiving feedback from other agents. We define:
    \begin{definition}[feedback loop]
         We say that there is a {\em feedback loop} between nodes $\node{i,t_1}$ and $\node{i,t_2}$ in run $r$ (and write $\node{i,t_1}\fb_r\node{i,t_2}$) if there is a node $\node{b,t}$ with $b\notin\{i,d_i\}$ such that $\node{i,t_1}\ob_r\node{b,t}\ob_r\node{i,t_2}$.
         For an operation $\X$, we say that $\X$ contains a feedback loop in $r$ if $\X.s\fb_r\X.e$.
        
    \end{definition}
Clearly, \Cref{lem:obcomposedOnlyOf_i_And_jit1} and \Cref{lem:obcomposedOnlyOf_i_And_jit3} in \Cref{lem:obcomposedOnlyOf_i_And_j} imply the necessity of $\RMW$ or $\F$. Using \Cref{thm:feedbackloop}, we show that \Cref{lem:obcomposedOnlyOf_i_And_jit2} also typically leads to the same requirement. In particular, we show that without  performing a~$\F$ or a $\RMW$, a process that runs solo cannot be guaranteed that an item it wrote or read from buffer has been propagated. More formally: 
    \begin{theorem}\label{thm:feedbackloop}
        Let $\X$ be an operation in $r$ that does not contain a feedback loop.
        Then there exists a run $\tilde{r}\loc r$ such that
            $\X$ runs solo
            in $\tilde{r}$.
    \end{theorem}
    
    \begin{proof}
    Let $\X$ be an operation in $r$ that does not contain a feedback loop. Denote $\X.s = \node{i,t_1}$ and $\X.e = \node{i,t_2}$.
        Let \[S_1\triangleq\Past^+(\X.e)\cup\Past^+(\{\node{b,t_1-1}\}_{b\in\Ag})\]

        We apply \Cref{thm:DtFTSO} to the run $r$, the set $S_1$ and $\Delta=t_2-t_1+1$ to delay all nodes not in $\Past^+(S_1)$ by $\Delta$, obtaining a run $r'\loc r$ (See \Cref{fig:runningsolo} (b)). By construction, every agent $b$ is delayed in~$r'$  for $\Delta$ rounds starting from time $t_b$, where $t_b$ is the latest time such that $\node{b,t_b} \ob_r \X.e$ (if no such $t_b$ exists, we set $t_b \triangleq t_1-1$). Observe that the only nodes at which actions are performed concurrent  with~$\X$ in~$r'$ are ones that are in $\Past_{r'}(S_1)$. 
        We now define:
\[S_2~\triangleq ~S_1\backslash\{\node{i,t},\node{d_i,t}|t\geq t_1-1\}\] 
That is, we obtain $S_2$ by  removing  from $S_1$ all nodes at $d_i$ or $i$ during $\X$'s execution. 
We claim that $\X.s \notin \Past_{r'}^+(S_2)$. Indeed, if $\X.s \ob_{r'} \theta$ for some $\theta = \node{b,t} \in S_2$, then $t>t_1$ and by definition of $S_1$, it follows that $\theta\ob_{r'}\X.e$ and $b \notin\{d_i, i\}$ (since nodes of $d_i,i$ at time greater than $t_1-1$ are not in $S_2$), implying a feedback loop $\X.s \ob_{r'} \theta \ob_{r'} \X.e$, contradicting the assumption that~$\X$ does not contain a feedback loop. 
We now apply \Cref{thm:DtFTSO} to $r'$ and $S_2$ with the same $\Delta=t_2-t_1+1$ and obtain the run $\tilde{r}$. Because $\X.s\notin\Past_{r'}^+(S_2)$, we now obtain that $\X$ runs solo  in~$\tilde{r}$ between times $t_1+\Delta$ and $t_2+\Delta$
(see \Cref{fig:runningsolo}(d)).
        \begin{figure}
            \centering
            \includegraphics[width=1\linewidth]{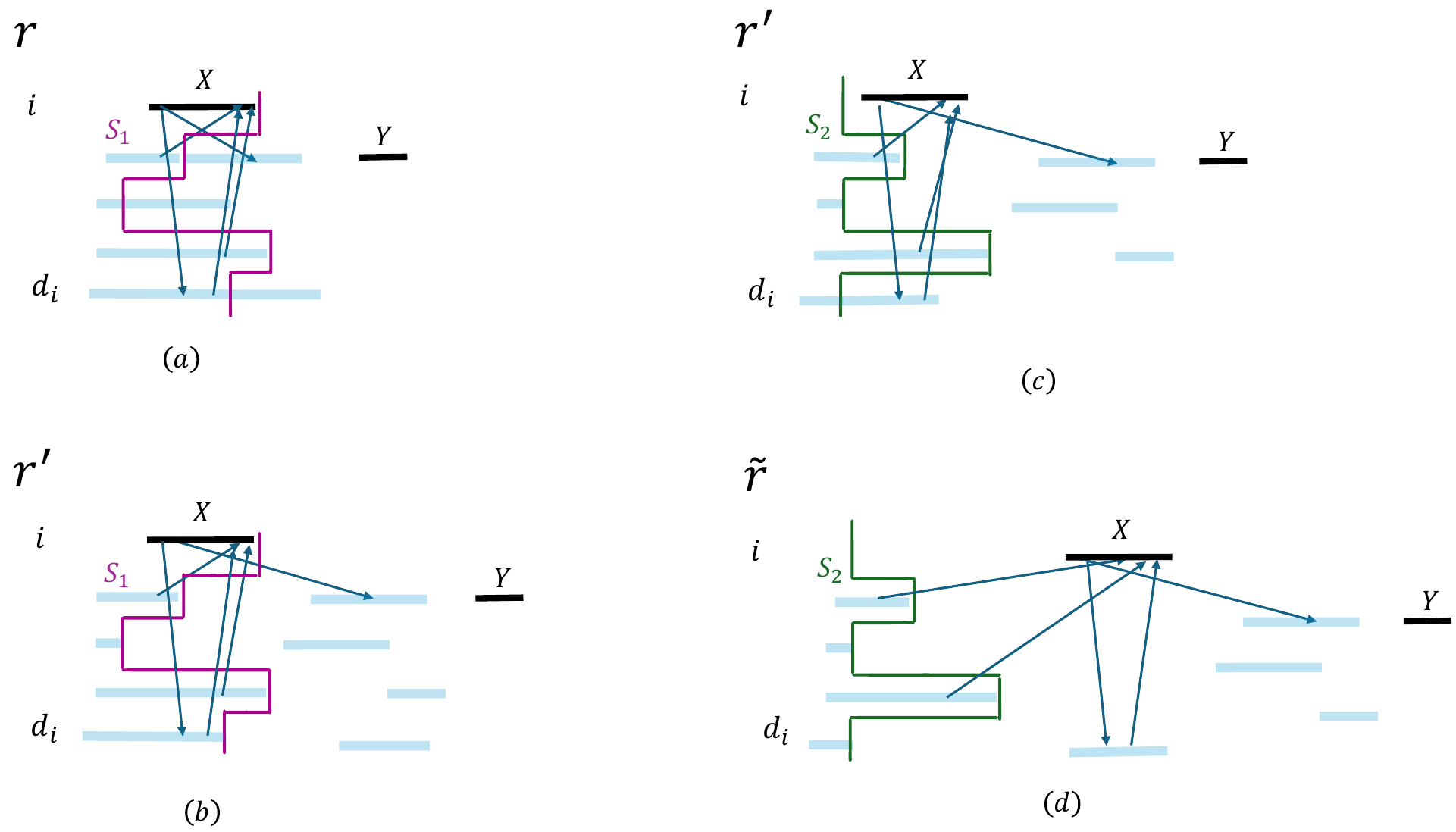}
            \caption{The construction of $\tilde{r}$ in \Cref{thm:feedbackloop}; blue lines represent nodes where agents move.}
            \label{fig:runningsolo}
        \end{figure}
    \end{proof}

We can now formally prove that in order to ensure that some value written in an operation has been propagated, the operation forces a process to either complete a feedback loop, or perform a costly synchronization action:
\begin{theorem}
\label{thm:feedbackorsynch}
 Let $\X$ be an operation in $r$ during which process $i$ writes a tag $\kappa$. If 
 \begin{enumerate}
     \item $\X$ does not contain a feedback loop, and 
     \item $i$ does not perform a $\F$ or $\RMW$ action during~$\X$. 
 \end{enumerate}
Then there is a run $r'\loc r$ in which $\kappa$ is not propagated to memory before $\X.e$.
\end{theorem}
\begin{proof}
    Let $r$ be a run in which process~$i$ runs solo between times $t_{\X.s}$ and $t_{\X.e}$ and in which $\node{i,t'}.\alpha=\W$ with $\node{i,t'}.\alpha.tag={\kappa}$ for some $t_{\X.s}\leq t'< t_{\X.e} $. Assume items 1 and 2 do not hold. By \Cref{thm:feedbackloop} we have that there is a run $r'\loc r$ in which $i$ runs solo during~$X$. 
        By local equivalence of $r'$, we have that $i$ writes $\kappa$ during $\X$ as well and $i$ does not perform $\F$ nor $\RMW$ during $\X$ in $r'$ neither. Thus, since $\X$ runs solo in $r'$, the only agents that might take actions during $\X$ are $i$ and $d_i$. By \Cref{obs:si-iimpliesfence} and the fact that $\X$ does not contain $\RMW$ or $\F$ we obtain that $\node{d_i,t_{\X.s}}\not\ob_r\X.e$. 
    Applying \Cref{thm:DtFTSO} to $r'$, with $S=\{\X.e\}\cup\{b,t_{\X.s}(r')\}_{b\in\Ag}$ and $\Delta=t_{\X.e}(r')-t_{\X.s}(r')$, we obtain a run $\tilde{r}\loc r'\loc r$ in which $\kappa$ is written by~$i$ during~$\X$ and is not propagated to memory before $\X.e$. Q.E.D.
\end{proof}

This is a powerful result. In memory models like TSO, a process can share information with others only by writing to memory and ensuring that the write becomes visible. Consequently, if completing an operation requires other processes to observe a write, the process must either construct a feedback loop—i.e., receive confirmation of information flow through another process, or perform a $\F$ or $\RMW$ operation. However, in obstruction-free protocols—where each process must be able to complete its operation independently of others—waiting for feedback is not an option. Our theorem thus shows that in such settings, using fences or $\RMW$ operations becomes a necessary mechanism to ensure the visibility of writes before an operation can safely terminate.
\subsection{Linearizable Objects}
\label{sec:linearizable}
We now show how our results can be applied in the context of linearizable implementation of shared objects imply the necessity of using $\F$ or $\RMW$. Roughly speaking, an implementation $\mathcal{I}$ of an object is said to be linearizable if in every execution of $\mathcal{I}$, operations appear to occur instantaneously in a way that is consistent with the original execution and that satisfy the sequential specification of the object. See \cite{HerlihyLineari} for a formal definition. We focus on the implementation of classical shared objects such as registers and snapshots in the TSO memory model.
\subsubsection{Registers}
A \emph{register} is a shared object that supports two operations: $\lREAD$ and $\lWRITE$. We consider the implementation of a multi-writer multi-reader (MWMR) register, where every process may perform both reads and writes, in the TSO memory model. A register is said to be \emph{atomic} if its operations appear instantaneous, and each read returns the value written by the most recent preceding write (or a fixed default value if no such write exists). For ease of exposition, we assume that a given value~$v$ will be written to the register at most once in any given execution. 
This avoids ambiguity in matching read operations to the corresponding write, which is particularly useful when reasoning about linearizability. Since our focus is on the ordering and visibility of operations, rather than value repetition or reuse, this restriction allows us to uniquely associate each read with a specific write without requiring additional disambiguation mechanisms such as timestamps or version numbers. Importantly, this assumption does not restrict generality, as any implementation where values may be written multiple times can be transformed into an equivalent one using uniquely tagged values. We consider in \Cref{thm:xaybsolonew} and \Cref{cor:writemustsync} register implementations that are obstruction-free \cite{DoubleQueueHerlihy2003}. Informally, an implementation is obstruction-free if it guarantees that a process is able to complete its pending operations in a finite number of its own steps, given that other processes do not take steps.
We say that an operation $\X$ is a \emph{$v$-operation}
and write $\X v$ if (i) $\X$ is a read that returns value $v$, or (ii) $\X$ is a write operation writing~$v$.

An important implication of \Cref{thm:DtFTSO} is the following result:
\begin{theorem}\label{thm:registerTSO}
    Let~$r$ be a run of a linearizable register implementation in TSO and let $\Xa <_r \Yb$ such that $\Yb$ completes in $r$. 
    If $\Yb$ runs in isolation in $r$ and $\sa \neq \sbb$, then $\Xa \ob_r \Yb$.
\end{theorem}
\begin{proofsketch}
    This result makes nontrivial use of \Cref{thm:DtFTSO}. An analogous result with happens before replacing 
    $\ob$ is Corollary~20 from~\cite{CommRequirDISC24}. Remarkably, the proof of Corollary~20 there applies verbatim to our statement, because it depends only on their Delaying the Future theorem, which is the counterpart of our \Cref{thm:DtFTSO} when \emph{occurs before} is replaced by \emph{happens before}.\footnote{Observe that the results of \cite{CommRequirDISC24} requires only obstruction-freedom although wait-freedom is assumed there.}
\end{proofsketch}

While one might expect that a read of a value would need to depend on the write of the same value, \Cref{thm:registerTSO} is stronger. For example, both $\Xa$ and $\Yb$ may be reads, or both could be write operations. Obviously, between a read of~$\sa$ and a read of $\sbb\ne\sa$ there must be a write of~$\sbb$. But the fact that an $\ob_r$ chain must be constructed between the read of~$\sa$ and the read of $\sbb$ formally implies that the implementation must use some shared variables in common to the both read operations.

\begin{theorem}\label{thm:xaybsolonew}
Let $\mathcal{I}$ be an obstruction-free linearizable register implementation in TSO, and let $\sbb\neq \sa $.
Suppose that an  operation $\Xa$ completes after running solo and in isolation in a run $r$ of $\mathcal{I}$. If $\Xa$ contains neither $\F$ nor $\RMW$ actions in $r$,  then 
 there exists a run $\tilde{r}$ of $\mathcal{I}$ indistinguishable to $i$ from~$r$ until the end of $\Xa$ in which $\Xa<_{\tilde{r}}\W(\sbb)$ for the 
 $\lWRITE$ operation $\W(\sbb)$, 
 and $\W(\sbb)$  contains a $\F$ or $\RMW$ action.
\end{theorem}
\begin{proof}
    Let $r$ be a run of $\mathcal{I}$ in which operation $\Xa$ performed by process $i$ runs solo and completes, and such that $\Xa$ contains neither $\F$ nor $\RMW$ actions in $r$. Because $\Xa$ runs solo, we have by \Cref{obs:si-iimpliesfence} that $\node{d_i,t_{\Xa.s}}\notin\Past^+(\Xa.e)$.
    Therefore, applying \Cref{thm:DtFTSO} to $r$, with~$S=\{\Xa.e\}\cup\{\node{b,t_{\Xa.s}}\}_{b\in\Ag}$ and $\Delta=t_{\Xa.e}-t_{\Xa.s}+2$ generates a run $r'\loc r$ in which $\node{d_i,t'}.\alpha\neq\prop$ for every $t_{\Xa.s}\leq t'\leq t_{\Xa.e}$ and clearly, $\Xa$ runs solo in $r'$ also. We define $\tilde{r}$ in the following way: $\tilde{r}[0,\dots,t_{\Xa.e}]\triangleq r'[0,\dots,t_{\Xa.e}]$. Now, we add the invocation of $\W(\sbb)$ at process $j$ at $t_{\Xa.e}+1$ in~$\tilde{r}$. In addition, we extend $\tilde{r}$ by letting only $j$ take steps and no propagate events by $d_i$ occur. Since $\mathcal{I}$ is obstruction-free, we have that $\Yb$ completes. Observe that by the construction, every occurs-before chain from $\Xa.s$ to $\W(\sbb).e$ is $\ijj$-only.
    Clearly, \cref{lem:obcomposedOnlyOf_i_And_jit1,lem:obcomposedOnlyOf_i_And_jit2} of \Cref{lem:obcomposedOnlyOf_i_And_j} wrt. $\Xa.s$ and $\W(\sbb)$ do not hold in $\tilde{r}$. Thus, \cref{lem:obcomposedOnlyOf_i_And_jit3} of \Cref{lem:obcomposedOnlyOf_i_And_j} must hold. Since $j$ does not take step between the beginning of $\Xa$ and the beginning of $\W(\sbb)$ in $\tilde{r}$, we obtain that for \cref{lem:obcomposedOnlyOf_i_And_jit3} of \Cref{lem:obcomposedOnlyOf_i_And_j} to hold, $\W(\sbb)$ must contain a Fence or~$\RMW$ in $\tilde{r}$. Otherwise, we obtain a contradiction to \Cref{thm:registerTSO}. 
\end{proof}
\begin{remark*}
    The requirement that $\Xa$ runs in isolation can be weakened to requiring that the transitive closure of the operations concurrent with $\Xa$ completes in $r$. \Cref{thm:xaybsolonew} assumes isolation to simplify the statement of the theorem.
\end{remark*}

Note that $\Xa$ and~$\Yb$ in this statement can both be $\lWRITE$ operations. 
 For obstruction-free register implementations,  techniques used in the proof of  
\Cref{thm:DtFTSO}  can be used to show:
\begin{corollary}\label{cor:writemustsync}
        Let $n\ge 2$ and let $\mathcal{I}$ be an obstruction-free linearizable implementation of a register in TSO for~$n$ processes. Then for every $m>0$ there must be a run $r_m$ of $\mathcal{I}$ in which exactly~$m$  $\lWRITE$ operations are performed, and each of the $\lWRITE$\hspace{.08mm}s
performs a fence~$\F$ or an $\RMW$. 
\end{corollary}
We remark that 
Castañeda \emph{et al.} prove in~\cite{WeakMemDISC24} that \emph{spec-available} implementations of linearizable registers in TSO must contain a run in which at least one $\lREAD$ or $\lWRITE$ operation performs an $\F$ or $\RMW$ action. Since every obstruction-free implementation is in particular spec-available, the same follows for obstruction-free implementations.
For obstruction-free implementations, \Cref{cor:writemustsync} provides sharper bounds. 
\subsubsection{Snapshots}
Consider a (single-writer) object $\snapshoto$ storing a
vector of a length~$\size{\Proc}$ over a set of values $W$
(also represented as function in $\Proc \to W$)
 with the initial vector of $\tup{\bot \til \bot}$. 
The operations are
$\set{\updateopp{w} \stt w\in V}$ and $\scanop$, 
with return values $\ack$ and $ (\Proc \to W)$, respectively. 
The specification of $\snapshoto$ 
consists of all complete sequential histories where each $\scanop$ event returns $\vec{V}$
such that $\vec{V}(\proc)$ is the value written by the last preceding $\updateop$ operation by process $\proc$, 
or $\bot$ if no such $\updateop$ exists.  

Similar to our definition of tags for write, read and propagate actions in TSO, we associate with each update operation and scanned value a $tag$ field as follows:
The $k$'th update  operation $\updateop$ by a process~$i$ in a given run will have the field $\updateop.tag\triangleq\node{i,k}$. Similarly, we associate a $tag$ field with each  vector component returned by $\scanop$ if the component of $i$ returned by $\scanop$ is $i$'s $k^{th}$ update operation in $r$. Therefore, the $tag$ field of a scan operation is a vector of length $|\Pi|$ of pairs of the form $\node{i,k}$.
Given two values $v,w$ in vectors returned by scan operations, we write $v\preceq  w$ if the tag associated with the value $v$ is smaller or equal to the one of $w$. 
To facilitate the analysis and w.l.o.g., we assume that a value is written at most once by a given process.
We can now show an unconditional occurs-before connection between scans and updates that precede each other: 
\begin{theorem}[update to scan]\label{thm:uptoscan}
   Let~$r$ be a run of a linearizable snapshot implementation. 
   \begin{enumerate}[label=(\alph*)]
       \item If $\updateop_i(v)<_r\scanop_j$ and $\scanop_j$ completes in~$r$, then $\updateop_i(v)\ob_r\scanop_j$, and 
       \item If $\scanop_i<_r\updateop_j(v)$ and $\updateop_j(v)$ completes in~$r$, then \mbox{$\scanop_i\ob_r\updateop_j(v)$}.
   \end{enumerate}
\end{theorem}
\begin{proof}
 To show part (a), let $r$ be a run such that $\updateop_i(v)<_r\scanop_j$ and denote by $\vec{V}$ the return value of $\scanop_j$. By the sequential specification of a snapshot object, we have that $\vec{V}[i]\succeq v$. Assume by way of contradiction that $\updateop_i(v)\not\ob_r\scanop_j$, i.e., for all $v'\succeq v$, it is the case that $\updateop_i(v').s\notin\Past^+(\{\scanop_j.e\})$.
    Thus, applying \Cref{thm:DtFTSO} on $r$, $S=\{\scanop_j.e\}$ and $\Delta=t_{\scanop_j.e}-t_{\updateop_i.s}+1$, we obtain a run $r'\loc r$ such that $\scanop_j<_{r'}\updateop_i(v')$ for all $v'\succeq v$ and $\scanop_j$ returns $\vec{V}[i]\succeq v$ in $r'$ as well, contradicting the snapshot specification.

The proof of (b) is similar: Let $r$ be a run such that $\scanop_i<_r\updateop_j(v)$. Denote by $\vec{V}$ the value returned by $\scanop_j$. By the sequential specification of a snapshot object, we have that $\vec{V}[j]\prec v$. Assume by way of contradiction that $\scanop_i\not\ob_r\updateop_j(v)$, i.e., $\scanop_i\notin\Past^+(\{\updateop_j(v).e\})$.
    Thus, applying \Cref{thm:DtFTSO} on $r$, $S=\{\updateop_j(v).e\}$ and $\Delta=t_{\updateop_j.e}-t_{\scanop_i.s}+1$, we obtain a run $r'\loc r$ such that 
    $\updateop_j(v)<_{r'}\scanop_i$ while $\scanop_i$ returns the same value $\vec{V}$ as in $r$, with $\vec{V}[j]\prec v$, contradicting the snapshot specification.
\end{proof}

What we showed here is the counterpart of \Cref{thm:registerTSO} for snapshot. Now, as in the proof of \Cref{thm:xaybsolonew}, we can use the connection between occurs before and synchronization established in \Cref{sec:applyDtF} to show:

\begin{theorem}\label{thm:scantoupnew}
    Let $\mathcal{I}$ be an obstruction-free linearizable snapshot implementation in TSO.
    Let $r$ be a run of $\mathcal{I}$ in which operation $\updateop_i(v)$ is performed by process $i$ running solo and completes. If $\updateop_i(v)$ contains neither $\F$ nor $\RMW$ actions in $r$, then for all operations $\scanop_j$, there exists a run $\tilde{r}$ of $\mathcal{I}$ indistinguishable to $i$ until the end of $\updateop_i(v)$ in which $\updateop_i(v)<_{\tilde{r}}\scanop_j$ and $\scanop_j$ contains a $\F$ or $\RMW$ action.
\end{theorem}
\begin{proof}
    The proof of \Cref{thm:scantoupnew} is identical to the one of \Cref{thm:xaybsolonew} with $\Xa$ being replaced by $\updateop_i(v)$ and $\Yb$ being replaced by $\scanop_j$.
\end{proof}
\begin{theorem}\label{thm:scanupdatenew}
    Let $\mathcal{I}$ be an obstruction-free linearizable snapshot implementation in TSO.
    Let $r$ be a run of $\mathcal{I}$ in which operation $\scanop_i$ is performed by process $i$ running solo and completes. If $\scanop_i$ contains neither $\F$ nor $\RMW$ actions in $r$, then for all operations $\updateop_j(v)$, there exists a run $\tilde{r}$ of $\mathcal{I}$ indistinguishable to $i$ until the end of $\scanop_i$ in which $\scanop_i<_{\tilde{r}}\updateop_j(v)$ and $\updateop_j(v)$ contains a $\F$ or $\RMW$ action.
\end{theorem}

\section{Discussion}
The register results of~\cite{WeakMemDISC24} exploit the fact that $\lREAD$ operations are one-sided non-commutative with respect to $\lWRITE$ operations, yielding (roughly) that there exist runs where at least one of two adjacent write–read operations must contain a $\F$ or $\RMW$.
Our framework strengthens this: we prove that there exist runs where 
$\lWRITE$\hspace{.12mm}s must contain a $\F$ or~$\RMW$.
In typical register implementations, $\lWRITE$\hspace{.12mm}s complete with an acknowledgment independently of preceding $\lWRITE$\hspace{.12mm}s, meaning they are not one-sided non-commutative.
Nevertheless, our tools show that there are runs where $\lWRITE$\hspace{.12mm}s must contain $\F$ or $\RMW$.
In addition, the DtF theorem and its uses, both in the TSO setting and in~\cite{CommRequirDISC24}, raise the question of whether similar results hold in other memory models.
Finally, we observe that although our $\ob$ relation in TSO is analogous to happens-before in message passing, it differs in that $\ob$ does not necessarily convey information, even under full information, whereas happens-before does.
Nevertheless, like happens-before, the $\ob$ relation remains a necessary condition for ordering in many cases. This fundamental difference is both interesting and worth deeper investigation.

  \bibliography{ref,z1,z2_DISC23,ref2,references}

  \appendix

\end{document}